\documentclass[a4paper,11pt]{amsart}

\usepackage[margin=1in,includehead,includefoot]{geometry}
\usepackage[OT4]{fontenc}
\usepackage{amsthm,amsmath,amssymb}
\usepackage{cmtt}
\usepackage{mathrsfs}
\usepackage{enumitem}
\usepackage{tikz}
\usetikzlibrary{decorations.pathreplacing}
\usepackage[normalsize]{subfigure}
\usepackage{caption}
\usepackage[pdftitle={Coloring triangle-free rectangle overlap graphs with O(log log n) colors},
            pdfauthor={Tomasz Krawczyk, Arkadiusz Pawlik, Bartosz Walczak},
            colorlinks=true,linkcolor=black,citecolor=black,filecolor=black,urlcolor=black]{hyperref}

\newtheorem{theorem}{Theorem}[section]
\newtheorem{lemma}[theorem]{Lemma}
\newtheorem{proposition}[theorem]{Proposition}
\newtheorem{corollary}[theorem]{Corollary}
\theoremstyle{definition}
\newtheorem{problem}{Problem}

\newenvironment{enumeratei}{\begin{enumerate}[label=\textup{(\roman*)}, noitemsep, topsep=1.5mm plus 1.5mm, leftmargin=*, widest=iii]}{\end{enumerate}}
\newenvironment{enumeratea}{\begin{enumerate}[label=\textup{(\arabic*)}, noitemsep, topsep=0pt, leftmargin=*]}{\end{enumerate}}
\newenvironment{enumeratex}[1]{\begin{enumerate}[label=\textup{(#1\arabic*)}, noitemsep, topsep=1.5mm plus 1.5mm, leftmargin=*]}{\end{enumerate}}
\newenvironment{enumeratey}[1]{\begin{enumerate}[label=\textup{(#1)}, noitemsep, topsep=1.5mm plus 1.5mm, leftmargin=*]}{\end{enumerate}}

\renewenvironment{itemize}{\begin{itemorig}[label=\textbullet, noitemsep, topsep=1.5mm plus 1.5mm, labelsep=.6em, labelindent=.4em, leftmargin=*]}{\end{itemorig}}

\newcommand{\cgC}{\mathcal{C}}
\newcommand{\cgF}{\mathcal{F}}
\newcommand{\cgI}{\mathcal{I}}
\newcommand{\cgJ}{\mathcal{J}}
\newcommand{\cgK}{\mathcal{K}}
\newcommand{\cgL}{\mathcal{L}}
\newcommand{\cgM}{\mathcal{M}}
\newcommand{\cgR}{\mathcal{R}}
\newcommand{\cgS}{\mathcal{S}}
\newcommand{\cgT}{\mathcal{T}}
\newcommand{\cgX}{\mathcal{X}}
\newcommand{\cgY}{\mathcal{Y}}
\newcommand{\cgZ}{\mathcal{Z}}

\newcommand{\scrP}{\mathscr{P}}

\newcommand{\setN}{\mathbb{N}}
\newcommand{\setR}{\mathbb{R}}

\newcommand{\nbhyphen}{\kern 0pt\hbox{-}\nobreak\hskip 0pt\relax}

\let\leq\leqslant
\let\geq\geqslant
\let\setminus\smallsetminus
\let\edgeto\rightarrow

\DeclareMathOperator{\rect}{rect}

\def\leftside{\ell}
\def\rightside{r}
\def\topside{t}
\def\bottomside{b}

\makeatletter
\let\old@setaddresses\@setaddresses
\def\@setaddresses{\bigskip\bgroup\parindent 0pt\let\scshape\relax\old@setaddresses\egroup}
\makeatother

\title[Coloring triangle-free rectangle overlap graphs with $O(\log\log n)$ colors]{\boldmath Coloring triangle-free rectangle overlap graphs with~$O(\log\log n)$ colors}

\author{Tomasz Krawczyk\and Arkadiusz Pawlik\and Bartosz Walczak}

\thanks{A journal version of this paper appeared in \emph{Discrete Comput.\ Geom.}, 53(1):199--220, 2015.
A preliminary version of this paper appeared as: Coloring triangle-free rectangular frame intersection graphs with $O(\log\log n)$ colors.\ In Andreas Brandst\"adt, Klaus Jansen, and R\"udiger Reischuk, editors, \emph{Graph-Theoretic Concepts in Computer Science (WG 2013)}, volume 8165 of \emph{Lecture Notes Comput.\ Sci.}, pages 333--344.\ Springer, Berlin, 2013.}
\thanks{Tomasz Krawczyk and Arkadiusz Pawlik were supported by Ministry of Science and Higher Education of Poland grant 884/N-ESF-EuroGIGA/10/2011/0 within ESF EuroGIGA project GraDR\@.
Bartosz Walczak was supported by Ministry of Science and Higher Education of Poland grant 884/N-ESF-EuroGIGA/10/2011/0 within ESF EuroGIGA project GraDR and by Swiss National Science Foundation grant 200020-144531.}

\address{Theoretical Computer Science Department, Faculty of Mathematics and Computer Science, Jagiellonian University, Krak\'ow, Poland}
\email{\mtt\{krawczyk,pawlik,walczak\mtt\}@tcs.uj.edu.pl}

\begin{document}

\begin{abstract}
Recently, it was proved that triangle-free intersection graphs of $n$ line segments in the plane can have chromatic number as large as $\Theta(\log\log n)$.
Essentially the same construction produces $\Theta(\log\log n)$-chromatic triangle-free intersection graphs of a variety of other geometric shapes---those belonging to any class of compact arc-connected sets in $\mathbb{R}^2$ closed under horizontal scaling, vertical scaling, and translation, except for axis-parallel rectangles.
We show that this construction is asymptotically optimal for intersection graphs of boundaries of axis-parallel rectangles, which can be alternatively described as overlap graphs of axis-parallel rectangles.
That is, we prove that triangle-free rectangle overlap graphs have chromatic number $O(\log\log n)$, improving on the previous bound of $O(\log n)$.
To this end, we exploit a relationship between off-line coloring of rectangle overlap graphs and on-line coloring of interval overlap graphs.
Our coloring method decomposes the graph into a bounded number of subgraphs with a tree-like structure that ``encodes'' strategies of the adversary in the on-line coloring problem.
Then, these subgraphs are colored with $O(\log\log n)$ colors using a combination of techniques from on-line algorithms (first-fit) and data structure design (heavy-light decomposition).
\end{abstract}

\maketitle

\section{Introduction}

A \emph{proper coloring} of a graph is an assignment of colors to the vertices of the graph such that no two adjacent ones are in the same color.
The minimum number of colors sufficient to color a graph $G$ properly is called the \emph{chromatic number} of $G$ and denoted by $\chi(G)$.
The maximum size of a clique (a set of pairwise adjacent vertices) in a graph $G$ is called the \emph{clique number} of $G$ and denoted by $\omega(G)$.
It is clear that $\chi(G)\geq\omega(G)$.
Classes of graphs for which there is a function $f\colon\setN\to\setN$ such that $\chi(G)\leq f(\omega(G))$ holds for any graph $G$ in the class are called \emph{$\chi$-bounded}.
A graph is \emph{triangle-free} if it does not contain a triangle, that is, if $\omega(G)\leq 2$.

It was observed in the 1940s that large cliques are not necessary for the chromatic number to grow.
Various classical constructions \cite{Myc55,Zyk49} show that it can be arbitrarily large even for triangle-free graphs.
Kim \cite{Kim95} constructed triangle-free graphs on $n$ vertices with chromatic number $\Theta(\sqrt{n/\log n})$.
An earlier result due to Ajtai, Koml\'os and Szemer\'edi \cite{AKS80} implies that this bound is tight.

In this paper, we focus on the relation between the chromatic number and the number of vertices in classes of triangle-free graphs arising from geometry.
The \emph{intersection graph} of a family of sets $\cgF$ is the graph with vertex set $\cgF$ and edge set consisting of pairs of intersecting members of $\cgF$.
The \emph{overlap graph} of a family of sets $\cgF$ is the graph with vertex set $\cgF$ and edge set consisting of pairs of members of $\cgF$ that intersect but neither contains the other.
We always assume that the family of sets $\cgF$ is finite.

It is well known that interval graphs (intersection graphs of intervals in $\setR$) are perfect---they satisfy $\chi(G)=\omega(G)$.
The study of the chromatic number of graphs with geometric representation was initiated in a seminal paper of Asplund and Gr\"unbaum \cite{AG60}, where they proved that the intersection graphs of families of axis-parallel rectangles in the plane are $\chi$-bounded.
In particular, they proved the tight bound of $6$ on the chromatic number of triangle-free intersection graphs of axis-parallel rectangles.
Gy\'arf\'as \cite{Gya85,Gya86} proved that the class of overlap graphs of intervals in $\setR$ is $\chi$-bounded.
By contrast, Burling \cite{Bur65} showed that triangle-free intersection graphs of axis-parallel boxes in $\setR^3$ can have arbitrarily large chromatic number.
Pawlik et al.\ \cite{PKK+13,PKK+14} provided a construction of triangle-free intersection graphs of segments, and more generally, triangle-free intersection graphs of families of vertically and horizontally scaled translates of any fixed arc-connected compact set in $\setR^2$ that is not an axis-parallel rectangle, with arbitrarily large chromatic number.
These graphs require $\Omega(\log\log n)$ colors, where $n$ is the number of vertices.
One of the problems posed in \cite{PKK+14} is to determine (asymptotically) the maximum chromatic number that a triangle-free intersection graph of $n$ segments can have.

We solve the analogous problem for triangle-free intersection graphs of \emph{frames}, which are boundaries of axis-parallel rectangles.
These graphs can be alternatively defined as overlap graphs of axis-parallel rectangles.
Therefore, they can be considered as two-dimensional generalizations of interval overlap graphs.
We show that the construction of Pawlik et al.\ is asymptotically best possible for these graphs.

\begin{theorem}
\label{thm:loglog}
Every triangle-free overlap graph of\/ $n$ rectangles (intersection graph of\/ $n$ frames) can be properly colored with\/ $O(\log\log n)$ colors.
\end{theorem}

For the completeness of exposition, we also include the proof of the bound from the other side, which appears in \cite{PKK+13}.

\begin{theorem}
\label{thm:construction}
There are triangle-free overlap graphs of\/ $n$ rectangles (intersection graphs of\/ $n$ frames) with chromatic number\/ $\Theta(\log\log n)$.
\end{theorem}

Note the difference in the behavior of rectangle intersection graphs and rectangle overlap graphs.
The former have chromatic number bounded by a function of their clique number.
The latter can have arbitrarily large chromatic number even when they are triangle-free.

Theorem \ref{thm:loglog} provides the first asymptotically tight bound on the chromatic number for a natural class of geometric intersection or overlap graphs that does not allow a constant bound.
So far, best upper bounds were of order $O(\log n)$, following from the results of McGuinness \cite{McG00} and Suk \cite{Suk14} for intersection graphs of families of shapes including segments and frames, or polylogarithmic in $n$, obtained by Fox and Pach \cite{FP14} for arbitrary families of curves with bounded clique number.
The only known lower bounds follow from the above-mentioned constructions of Burling and Pawlik et al.
We hope that our ideas will lead to improving the bounds for other important classes, in particular, for segment intersection graphs.

\emph{On-line coloring} is an extensively studied variant of the coloring problem.
The difference between ordinary and on-line coloring is that in the on-line setting, the vertices are introduced one by one and the coloring algorithm must assign colors to them immediately, knowing only the edges between vertices presented thus far.
Our proof exploits a correspondence between on-line coloring of interval overlap graphs and ordinary (off-line) coloring of rectangle overlap graphs.
We obtain a structural decomposition of an arbitrary rectangle overlap graph with bounded clique number into a bounded number of so-called \emph{directed families} of rectangles.
For families whose overlap graphs are triangle-free, we can further decompose them into so-called \emph{clean families}, in which no rectangle is entirely contained in the intersection of two overlapping rectangles.
It turns out that overlap graphs of clean directed families of rectangles have a particular structure of what we call \emph{overlap game graphs}, that is, they can be viewed as encodings of adversary strategies in the on-line interval overlap graph coloring problem.
We succeed in coloring overlap game graphs with $O(\log\log n)$ colors by combining two ideas:\ heavy-light decomposition and first-fit coloring.
The reductions to directed and clean families are purely geometric---they exploit $\chi$-boundedness results and coloring techniques from \cite{AG60,Gya85,McG00,Suk14}.

\section{Overview}
\label{sec:overview}

All rectangles that we consider are axis-parallel, that is, their sides are parallel to the horizontal or the vertical axes.
Throughout the paper, we also assume that all rectangles are in general position, that is, no corner of any rectangle lies on the boundary of another rectangle.
We can easily adjust any family of rectangles so as to satisfy this condition without changing the overlap relation, just by expanding each rectangle in every direction by a tiny amount inversely proportional to the area of the rectangle.
The boundary of a rectangle is a \emph{frame}.
The \emph{filling rectangle} of a frame $F$, denoted by $\rect(F)$, is the rectangle whose boundary is $F$.

From now on, we will work with families of frames and their intersection graphs.
We denote the chromatic and the clique numbers of the intersection graph of a family of frames $\cgF$ by $\chi(\cgF)$ and $\omega(\cgF)$, respectively.
Triangles, cliques and connected components of the intersection graph of $\cgF$ are simply called triangles, cliques and connected components of $\cgF$.
The $x$-coordinates of the left and right sides and the $y$-coordinates of the bottom and top sides of a frame $F$ are denoted by $\leftside(F)$, $\rightside(F)$, $\bottomside(F)$, $\topside(F)$, respectively.
Thus $\rect(F)=[\leftside(F),\rightside(F)]\times[\bottomside(F),\topside(F)]$.

We distinguish the following types of frame intersections, illustrated in Figures \ref{fig:crossing}--\subref{fig:diagonal2}:\ \emph{crossings}, \emph{leftward\nbhyphen}, \emph{rightward\nbhyphen}, \emph{downward\nbhyphen} and \emph{upward-directed intersections}, and \emph{diagonal intersections}.
A family of frames $\cgF$ is \emph{leftward\nbhyphen}, \emph{rightward\nbhyphen}, \emph{downward\nbhyphen} or \emph{upward-directed} if the following condition is satisfied:
\begin{enumeratex}{F}
\item\label{cond:frame-directed} for any two intersecting frames in $\cgF$, their intersection is leftward\nbhyphen, rightward\nbhyphen, downward\nbhyphen{} or upward-directed, respectively.
\end{enumeratex}
A family of frames $\cgF$ is \emph{directed} if it is leftward\nbhyphen, rightward\nbhyphen, downward\nbhyphen{} or upward-directed.
Note that in a directed family, we still allow only one of the four types of directed intersections, we just do not specify which one.
A family of frames $\cgF$ is \emph{clean} if the following holds:
\begin{enumeratex}{F}
\setcounter{enumi}{1}
\item\label{cond:frame-clean} no frame in $\cgF$ is enclosed in two intersecting frames in $\cgF$ (see Figure \ref{fig:clean}).
\end{enumeratex}

\begin{figure}[t]
\centering
\subfigure[\unskip]{\label{fig:crossing}\begin{tikzpicture}[scale=0.8]
  \draw (0.0,0.4) rectangle (1.5,1.1);
  \draw (0.4,0.0) rectangle (1.1,1.5);
\end{tikzpicture}}
\hspace{4mm}
\subfigure[\unskip]{\label{fig:leftward}\begin{tikzpicture}[scale=0.8]
  \draw (0.0,0.4) rectangle (1.0,1.1);
  \draw (0.5,0.0) rectangle (1.5,1.5);
\end{tikzpicture}}
\hspace{4mm}
\subfigure[\unskip]{\begin{tikzpicture}[scale=0.8]
  \draw (0.5,0.4) rectangle (1.5,1.1);
  \draw (0.0,0.0) rectangle (1.0,1.5);
\end{tikzpicture}}
\hspace{4mm}
\subfigure[\unskip]{\begin{tikzpicture}[scale=0.8]
  \draw (0.4,0.0) rectangle (1.1,1.0);
  \draw (0.0,0.5) rectangle (1.5,1.5);
\end{tikzpicture}}
\hspace{4mm}
\subfigure[\unskip]{\label{fig:upward}\begin{tikzpicture}[scale=0.8]
  \draw (0.4,0.5) rectangle (1.1,1.5);
  \draw (0.0,0.0) rectangle (1.5,1.0);
\end{tikzpicture}}
\hspace{4mm}
\subfigure[\unskip]{\label{fig:diagonal1}\begin{tikzpicture}[scale=0.8]
  \draw (0.0,0.0) rectangle (1.0,1.0);
  \draw (0.5,0.5) rectangle (1.5,1.5);
\end{tikzpicture}}
\hspace{4mm}
\subfigure[\unskip]{\label{fig:diagonal2}\begin{tikzpicture}[scale=0.8]
  \draw (0.5,0.0) rectangle (1.5,1.0);
  \draw (0.0,0.5) rectangle (1.0,1.5);
\end{tikzpicture}}
\hspace{4mm}
\subfigure[\unskip]{\label{fig:clean}\begin{tikzpicture}[scale=0.8]
  \draw (0.5,0.5) rectangle (1.25,1.0);
  \draw (0.25,0.25) rectangle (1.75,1.25);
  \draw (0.0,0.0) rectangle (1.5,1.5);
\end{tikzpicture}}
\caption{\subref{fig:crossing} crossing; \subref{fig:leftward}--\subref{fig:upward} leftward\nbhyphen, rightward\nbhyphen, downward\nbhyphen{} and upward-directed intersections, respectively; \subref{fig:diagonal1}--\subref{fig:diagonal2} diagonal intersections; \subref{fig:clean} forbidden configuration in a clean rightward-directed family}
\end{figure}
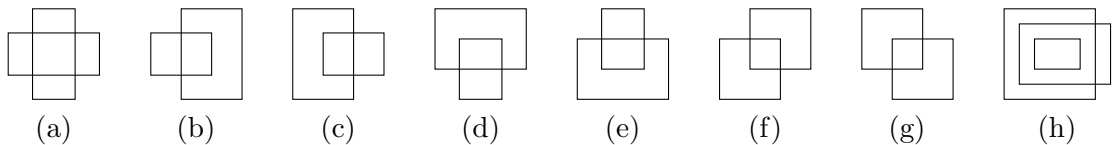

The first step in our proof of Theorem \ref{thm:loglog} is to reduce the problem of coloring arbitrary triangle-free families of frames to the problem of coloring clean directed triangle-free families of frames.
This is done by the following two lemmas, the first of which works with any bound on the clique number, not only for triangle-free families.

\begin{lemma}
\label{lem:frame-geometry}
Every family of frames\/ $\cgF$ with\/ $\omega(\cgF)\leq k$ can be partitioned into a bounded number of directed subfamilies, where the bound depends only on\/ $k$.
\end{lemma}

\begin{lemma}
\label{lem:frame-clean}
Every triangle-free family of frames can be partitioned into two subfamilies so that every connected component of either subfamily is a clean family of frames.
\end{lemma}

\noindent
The proofs of both lemmas are technical and are postponed to Section \ref{sec:geometry}.

The next step is a more abstract description of the structure of intersection graphs of clean directed families of frames in terms of intervals in $\setR$.
The family of all closed intervals in $\setR$ is denoted by $\cgI$.
The left and right endpoints of an interval $I$ are denoted by $\leftside(I)$ and $\rightside(I)$, respectively.
Two intervals \emph{overlap} if and only if they intersect but neither contains the other.
The \emph{overlap graph} defined on a family of intervals has an edge for each pair of overlapping intervals.
Again, we can assume without loss of generality that the intervals representing an overlap graph are in general position, which means that all their endpoints are distinct.
With this assumption, intervals $I_1$ and $I_2$ overlap if and only if $\leftside(I_1)<\leftside(I_2)<\rightside(I_1)<\rightside(I_2)$ or $\leftside(I_2)<\leftside(I_1)<\rightside(I_2)<\rightside(I_1)$.

A \emph{rooted forest} is a forest in which every tree has one vertex chosen as its root.
Let $G$ be a graph, $M$ be a rooted forest with $V(M)=V(G)$, and $\mu\colon V(G)\to\cgI$.
For $u,v\in V(G)$, we write $u\prec v$ if $u\neq v$ and $u$ is an ancestor of $v$ in $M$.
The graph $G$ is an \emph{overlap game graph} with \emph{meta-forest} $M$ and \emph{representation} $\mu$ if the following conditions are satisfied:
\begin{enumeratex}{G}
\item\label{def:increasing} $\leftside(\mu(x))<\leftside(\mu(y))$ whenever $x\prec y$;
\item\label{def:game-graph} $xy\in E(G)$ if and only if $x\prec y$ or $y\prec x$ and $\mu(x)$ and $\mu(y)$ overlap;
\item\label{def:forbidden-structure} there are no $x$, $y$, $z$ such that $x\prec y\prec z$, $\mu(x)$ and $\mu(y)$ overlap, and $\mu(z)\subset\mu(x)\cap\mu(y)$.
\end{enumeratex}
In particular, for every set of vertices $Q$ on a common path from a root to a leaf in $M$, the induced subgraph $G[Q]$ of $G$ is isomorphic to the overlap graph of the family of intervals $\{\mu(x)\colon x\in Q\}$.
Here, once again, we assume without loss of generality that the intervals $\mu(x)$ for $x\in Q$ are in general position.
Since two vertices $x,y\in V(G)$ that are not in the ancestor-descendant relation of $M$ ($x\not\prec y$ and $y\not\prec x$) are never adjacent in $G$, the graph $G$ is the union of the overlap graphs $G[Q]$ over the vertex sets $Q$ of all paths from a root to a leaf in $M$.

\begin{lemma}
\label{lem:relation}
A graph is an overlap game graph if and only if it is isomorphic to the intersection graph of a clean directed family of frames.\footnotemark
\end{lemma}

\addtocounter{footnote}{-1}  
\addtocounter{Hfootnote}{-1} 

\begin{lemma}
\label{lem:game-graph-coloring}
Every triangle-free overlap game graph has chromatic number\/ $O(\log\log n)$.\footnotemark
\end{lemma}

\footnotetext{Lemmas \ref{lem:relation} and \ref{lem:game-graph-coloring} remain valid if we drop the cleanliness condition on the families of frames considered and the condition \ref{def:forbidden-structure} in the definition of an overlap game graph.
However, the condition \ref{def:forbidden-structure} is necessary for the proof of Lemma \ref{lem:game-graph-coloring}.
To derive the analogue of Lemma \ref{lem:game-graph-coloring} without the condition \ref{def:forbidden-structure}, we would first apply the analogue of Lemma \ref{lem:relation} to turn the graph into a directed family of frames, then apply Lemma \ref{lem:frame-clean} to partition it into two subfamilies with clean connected components, and then we would apply Lemma \ref{lem:relation} again to these connected components going back to overlap game graphs, for which we would apply the original Lemma \ref{lem:game-graph-coloring}.}

Now, Theorem \ref{thm:loglog} follows from Lemmas \ref{lem:frame-geometry} and \ref{lem:frame-clean}, the fact that each connected component of a graph can be colored separately, and Lemmas \ref{lem:relation} and \ref{lem:game-graph-coloring}.
Theorem \ref{thm:construction} follows from Lemma \ref{lem:relation} and the following result, implicit in \cite{PKK+13}, which asserts that the bound in Lemma \ref{lem:game-graph-coloring} is tight.

\begin{lemma}
\label{lem:construction}
There are triangle-free overlap game graphs with chromatic number\/ $\Omega(\log\log n)$.
\end{lemma}

We prove Lemma \ref{lem:relation} in Section \ref{sec:relation} and Lemma \ref{lem:game-graph-coloring} in Section \ref{sec:coloring}.
In Section \ref{sec:restricted}, we mostly recall some known results about on-line colorings of forests, which motivate the proof of Lemma \ref{lem:construction} from \cite{PKK+13} and the coloring algorithm presented in Section \ref{sec:coloring}.
We also include a sketch of the proof of Lemma \ref{lem:construction}, for completeness and to illustrate the idea behind overlap game graphs.

We complete this outline by explaining the meaning of the word ``game'' in the term ``overlap game graphs''.
Let $k\in\setN$.
Consider the following \emph{overlap game} between two players:\ Presenter, who builds a family of intervals presenting them one by one, and Algorithm, who colors them on-line, that is, each interval is assigned its color right after it is presented and without possibility of changing the color later.
Presenter's moves are restricted by the following rules:
\begin{enumeratex}{I}
\item\label{def:game-i} if an interval $I_2$ is presented after $I_1$, then $\leftside(I_1)<\leftside(I_2)$;
\item\label{def:game-ii} no three intervals $I_1$, $I_2$, $I_3$ such that $I_1$ and $I_2$ overlap and $I_3\subset I_1\cap I_2$ are ever presented.
\end{enumeratex}
The coloring constructed by Algorithm has to be a proper coloring of the overlap graph defined by the intervals presented.
Presenter aims to force Algorithm to use more than $k$ colors, while Algorithm tries to keep using at most $k$ colors.

Every finite strategy of Presenter gives rise to an overlap game graph $G$ with meta-forest $M$ and representation $\mu$ such that the root-to-leaf paths in $M$ correspond to the intervals presented in the possible scenarios of the strategy.
Specifically, each root $r$ of $M$ corresponds to an interval $\mu(r)$ that can be played in Presenter's first move, and each child of a vertex $x$ of $M$ corresponds to an interval that Presenter can play right after $\mu(x)$ at the position represented by $x$.
Conversely, an overlap game graph $G$ with meta-forest $M$ and representation $\mu$ can be interpreted as a non-deterministic strategy of Presenter, as follows.
Presenter starts with an arbitrarily chosen root $r$ of $M$ presenting $\mu(r)$, and then, in each move from a position $u$ in $M$, follows to an arbitrarily chosen child $v$ of $u$ presenting $\mu(v)$.
The key observation is that Algorithm has a strategy to use at most $k$ colors against such a strategy of Presenter if and only if $\chi(G)\leq k$.
The proof of Lemma \ref{lem:construction} constructs a strategy of Presenter forcing Algorithm to use more than $k$ colors by presenting at most $2^k$ intervals, while the proof of Lemma \ref{lem:game-graph-coloring} essentially shows that every such strategy needs to have a double exponential number of scenarios.

The presented ``on-line'' interpretation of overlap game graphs is exploited in the proof of Lemma \ref{lem:construction} presented in the next section.
It may also provide a useful insight into our arguments in Section \ref{sec:coloring}, which are formulated using the static definition of an overlap game graph.

\section{Restricted families of intervals}
\label{sec:restricted}

Let $\cgJ$ be a family of intervals in $\setR$ in general position.
We say that $\cgJ$ is \emph{restricted} if the following condition is satisfied:
\begin{enumeratey}{R}
\item\label{def:restricted} for any three intervals $I_1,I_2,I_3\in\cgJ$, if $\leftside(I_1)<\leftside(I_2)<\leftside(I_3)$ and $I_1$ and $I_2$ overlap, then $\rightside(I_1)<\leftside(I_3)$.
\end{enumeratey}
That is, restricted families exclude the three configurations of intervals presented in Figure \ref{fig:restricted}.

\begin{figure}[t]
\centering
\subfigure[\unskip]{\label{fig:restricted-clean}\begin{tikzpicture}[scale=0.35,xscale=1.2]
  \draw[|-|] (0,0)--(6,0);
  \draw[|-|] (1,1)--(7,1);
  \draw[|-|] (2,2)--(5,2);
\end{tikzpicture}}
\hspace{1cm}
\subfigure[\unskip]{\label{fig:restricted-secondary}\begin{tikzpicture}[scale=0.35,xscale=1.2]
  \draw[|-|] (0,0)--(5,0);
  \draw[|-|] (1,1)--(7,1);
  \draw[|-|] (2,2)--(6,2);
\end{tikzpicture}}
\hspace{1cm}
\subfigure[\unskip]{\label{fig:restricted-triangle}\begin{tikzpicture}[scale=0.35,xscale=1.2]
  \draw[|-|] (0,0)--(5,0);
  \draw[|-|] (1,1)--(6,1);
  \draw[|-|] (2,2)--(7,2);
\end{tikzpicture}}
\caption{Three configurations of intervals excluded in restricted families}
\label{fig:restricted}
\end{figure}
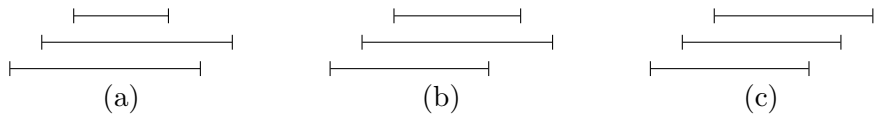

For $I_1,I_2\in\cgJ$, we write $I_1\edgeto I_2$ if $\leftside(I_1)<\leftside(I_2)<\rightside(I_1)<\rightside(I_2)$.
That is, $\edgeto$ is the orientation of the edges of the overlap graph of $\cgJ$ from left to right.

\begin{lemma}
\label{lem:restricted}
Let\/ $\cgJ$ be a restricted family of intervals.
\begin{enumeratea}
\item\label{lem:forest} For every\/ $I_1\in\cgJ$, there is at most one\/ $I_2\in\cgJ$ such that\/ $I_1\edgeto I_2$.
In particular, the overlap graph of\/ $\cgJ$ is a forest.
\item\label{lem:interlace} There are no\/ $I_1,I_2,I_3,I_4\in\cgJ$ such that\/ $\leftside(I_1)<\leftside(I_2)<\leftside(I_3)<\leftside(I_4)$, $I_1\edgeto I_3$, and\/ $I_2\edgeto I_4$.
\end{enumeratea}
\end{lemma}

\begin{proof}
Let $I_1,I_2,I_3\in\cgJ$, $\leftside(I_1)<\leftside(I_2)<\leftside(I_3)$, and $I_1\edgeto I_2$.
It follows from \ref{def:restricted} that $\rightside(I_1)<\leftside(I_3)$.
Hence the intervals $I_1$ and $I_3$ are disjoint, so $I_1\not\edgeto I_3$.
This completes the proof of \ref{lem:forest}.

For the proof of \ref{lem:interlace}, let $I_1,I_2,I_3,I_4\in\cgJ$, $\leftside(I_1)<\leftside(I_2)<\leftside(I_3)<\leftside(I_4)$, and $I_1\edgeto I_3$.
It follows from \ref{lem:forest} that $I_1\not\edgeto I_2\not\edgeto I_3$.
This and $I_1\edgeto I_3$ yields $\leftside(I_1)<\leftside(I_2)<\rightside(I_2)<\leftside(I_3)<\rightside(I_1)<\rightside(I_3)$.
Therefore, since $\leftside(I_3)<\leftside(I_4)$, the intervals $I_2$ and $I_4$ are disjoint, so $I_2\not\edgeto I_4$.
\end{proof}

The \emph{restricted overlap game} is a variant of the overlap game in which Presenter is required to present a restricted family of intervals.
Note that the condition that the family is restricted implies the condition \ref{def:game-ii} of the definition of the overlap game.

By Lemma \ref{lem:restricted} \ref{lem:forest}, the overlap graphs of families of intervals presented in the restricted overlap game are forests.
A well-known result in the area of on-line graph coloring algorithms asserts that Presenter has a strategy to force Algorithm to use more than $k$ colors on a forest with at most $2^k$ vertices.
Such a strategy was first found by Bean \cite{Bea76} and later rediscovered by Gy\'arf\'as and Lehel \cite{GL88}.
Erlebach and Fiala \cite{EF02} described a particular version of such a strategy, which can be carried out on forests represented as intersection graphs of discs or axis-parallel squares in the plane.
Pawlik et al.\ \cite{PKK+13} implemented the same strategy on forests represented as overlap graphs of intervals presented in the restricted overlap game.

We present this strategy first for abstract forests, and then in the restricted overlap game.
Next, we use it to prove Lemma \ref{lem:construction}.
Finally, we show that this strategy is optimal, that is, any forest with fewer than $2^k$ vertices can be properly colored on-line using at most $k$ colors.

\begin{proposition}[Bean \cite{Bea76}; Gy\'arf\'as, Lehel \cite{GL88}]
\label{prop:forest-strategy}
Presenter has a strategy to force Algorithm to use more than\/ $k$ colors by presenting a forest with at most\/ $2^k$ vertices.
\end{proposition}

\begin{proof}
For a tree $T$ of a forest presented on-line, let $r(T)$ denote the vertex of $T$ presented last.
The strategy constructs, by induction, a forest of trees $T_1,\ldots,T_m$ with at most $2^k-1$ vertices in total such that Algorithm is forced to use at least $k$ colors on the vertices $r(T_1),\ldots,r(T_m)$.

\begin{figure}[t]
\centering
\begin{tikzpicture}[xscale=0.7,>=latex,label distance=2pt]
  \node (l1) at (4,-1) {$r(T_1),\ldots,r(T_m)$};
  \node (l2) at (12,-1) {$r(T'_1),\ldots,r(T'_{\smash{m'}})$};
  \tikzstyle{every node}=[circle,minimum size=3pt,inner sep=0pt,draw,fill];
  \draw[double distance=8pt,line cap=round] (0,0)--(1,0) (2,0)--(3,0) (4,0)--(5,0) (6,0)--(7,0);
  \node (a) at (1,0) {};
  \node (b) at (3,0) {};
  \node (c) at (5,0) {};
  \node (d) at (7,0) {};
  \draw[double distance=8pt,line cap=round] (8,0)--(9,0) (10,0)--(11,0) (12,0)--(13,0) (14,0)--(15,0);
  \node (e) at (9,0) {};
  \node (f) at (11,0) {};
  \node (g) at (13,0) {};
  \node (h) at (15,0) {};
  \node[label=below:$v$] (i) at (16,0) {};
  \path (i) edge[bend right=30] (e) edge[bend right=30] (f) edge[bend right=30] (g) edge[bend right=30] (h);
  \path[->,dotted] (l1) edge (a) edge (b) edge (c) edge (d);
  \path[->,dotted] (l2) edge (e) edge (f) edge (g) edge (h);
\end{tikzpicture}
\caption{Illustration for the proof of Proposition \ref{prop:forest-strategy}}
\label{fig:forest-strategy}
\end{figure}
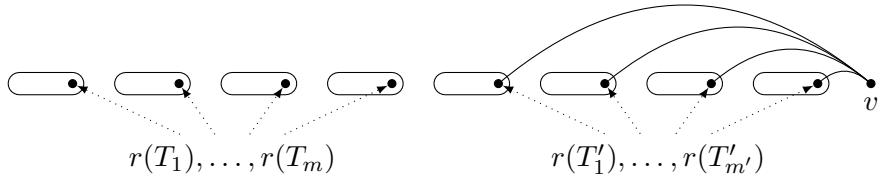

The case $k=1$ is trivial---it is enough to present a single vertex.
To obtain a strategy for $k\geq 2$, we apply the strategy for $k-1$ twice, building two forests of trees $T_1,\ldots,T_m$ and $T'_1,\ldots,T'_{\smash{m'}}$ with at most $2^k-2$ vertices in total such that Algorithm is forced to use at least $k-1$ colors on $r(T_1),\ldots,r(T_m)$ and at least $k-1$ colors on $r(T'_1),\ldots,r(T'_{\smash{m'}})$.
If the two sets of colors are different, then we have already forced Algorithm to use at least $k$ colors on $r(T_1),\ldots,r(T_m),r(T'_1),\ldots,r(T'_{\smash{m'}})$.
Otherwise, we present one more vertex $v$ and connect it to $r(T'_1),\ldots,r(T'_{\smash{m'}})$, thus merging $T'_1,\ldots,T'_{\smash{m'}}$ into one tree $T_{m+1}$ with $r(T_{m+1})=v$.
Algorithm has to color $v$ with a color different from those used on $r(T'_1),\ldots,r(T'_{\smash{m'}})$.
Hence it has been forced to use at least $k$ colors on $r(T_1),\ldots,r(T_{m+1})$.
See Figure \ref{fig:forest-strategy} for an illustration.

In the final step, after playing the strategy for $k$ claimed above, we present one more vertex $v$ connected to $r(T_1),\ldots,r(T_m)$, on which a $(k+1)$st color must be used.
This way, we have forced the use of more than $k$ colors on a tree with at most $2^k$ vertices.
\end{proof}

\begin{proposition}[Pawlik et al.\ \cite{PKK+13}]
\label{prop:overlap-strategy}
Presenter has a strategy to force Algorithm to use more than\/ $k$ colors on a family of at most\/ $2^k$ intervals in the restricted overlap game.
\end{proposition}

\begin{proof}
A \emph{tree} of a restricted family of intervals $\cgJ$ is a subfamily $\cgT\subset\cgJ$ that is a connected component of the overlap graph of $\cgJ$.
For a tree $\cgT$ of a family of intervals presented in the restricted overlap game, let $R(\cgT)$ denote the interval in $\cgT$ presented last.

The strategy is an adaptation of the strategy described in the proof of Proposition \ref{prop:forest-strategy}.
It constructs, by induction, a restricted family of intervals $\cgJ$ with trees $\cgT_1,\ldots,\cgT_m$ of total size at most $2^k-1$ such that
\begin{enumeratei}
\item\label{item:strategy-i} $R(\cgT_1)\supset\cdots\supset R(\cgT_m)$,
\item\label{item:strategy-ii} every interval in $\cgJ\setminus\{R(\cgT_1),\ldots,R(\cgT_m)\}$ lies entirely to the left of $\rightside(R(\cgT_m))$,
\item\label{item:strategy-iii} Algorithm is forced to use at least $k$ colors on the intervals $R(\cgT_1),\ldots,R(\cgT_m)$.
\end{enumeratei}

\begin{figure}[t]
\centering
\begin{tikzpicture}[scale=0.35,xscale=1.2]
  \draw[|-|] (-10,0)--(10,0);
  \draw[|-|] (-9,1)--(9,1);
  \draw[|-|] (-8,2)--(8,2);
  \draw[|-|] (-2,3)--(6,3);
  \draw[|-|] (-1,4)--(5,4);
  \draw[|-|] (0,5)--(4,5);
  \draw[|-|] (3,6)--(7,6);
  \draw[dotted] (11,0)--(10,0);
  \draw[dotted] (11,2)--(8,2);
  \draw[dotted] (11,3)--(6,3);
  \draw[dotted] (11,5)--(4,5);
  \begin{scope}[xshift=10pt]
  \draw[decorate,decoration=brace] (11,2)--(11,0);
  \node[inner sep=6pt,anchor=west] at (11,1) {$R(\cgT_1),\ldots,R(\cgT_m)$};
  \draw[decorate,decoration=brace] (11,5)--(11,3);
  \node[inner sep=6pt,anchor=west] at (11,4) {$R(\cgT'_1),\ldots,R(\cgT'_{\smash{m'}})$};
  \end{scope}
  \node[inner sep=2pt,anchor=south] at (5,6) {$J$};
  \draw[dashed] (-5,-0.4)--(-5,5.6);
  \draw[dashed] (1.5,-0.4)--(1.5,5.6);
  \node[inner sep=4pt,anchor=south] at (-5,5.4) {$x$};
  \node[inner sep=4pt,anchor=south] at (1.5,5.4) {$x'$};
\end{tikzpicture}
\caption{Illustration for the proof of Proposition \ref{prop:overlap-strategy}}
\label{fig:overlap-strategy}
\end{figure}
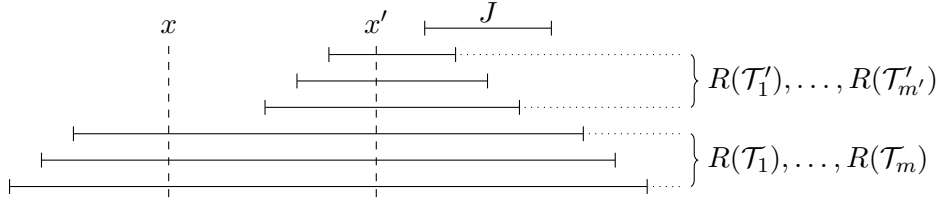

The case $k=1$ is trivial---it is enough to present a single interval.
To obtain a strategy for $k\geq 2$, we first apply the strategy for $k-1$, building a restricted family of intervals $\cgJ$ with trees $\cgT_1,\ldots,\cgT_m$ of total size at most $2^{k-1}-1$ such that \ref{item:strategy-i}--\ref{item:strategy-iii} are satisfied.
Let $x$ be the maximum of $\leftside(R(\cgT_m))$ and the right endpoints of all intervals in $\cgJ\setminus\{R(\cgT_1),\ldots,R(\cgT_m)\}$.
We apply the strategy for $k-1$ again, this time playing entirely inside the interval $(x,\rightside(R(\cgT_m)))$, building a restricted family of intervals $\cgJ'$ with trees $\cgT'_1,\ldots,\cgT'_{\smash{m'}}$ of total size at most $2^{k-1}-1$ such that \ref{item:strategy-i}--\ref{item:strategy-iii} are satisfied.
Let $x'$ be the maximum of $\leftside(R(\cgT'_{\smash{m'}}))$ and the right endpoints of all intervals in $\cgJ'\setminus\{R(\cgT'_1),\ldots,R(\cgT'_{\smash{m'}})\}$.
It is clear that the family $\cgJ\cup\cgJ'$ with trees $\cgT_1,\ldots,\cgT_m,\cgT'_1,\ldots,\cgT'_{\smash{m'}}$ satisfies \ref{item:strategy-i} and \ref{item:strategy-ii}.
If the two sets of at least $k-1$ colors used on $R(\cgT_1),\ldots,R(\cgT_m)$ and $R(\cgT'_1,\ldots,\cgT'_{\smash{m'}})$ differ, then we have already forced Algorithm to use at least $k$ colors on $R(\cgT'_1),\ldots,R(\cgT_m),R(\cgT'_1),\ldots,R(\cgT'_{\smash{m'}})$, so \ref{item:strategy-iii} is also satisfied.
Otherwise, we present one more interval $J$ with $\leftside(J)\in(x',\rightside(R(\cgT'_{\smash{m'}})))$ and $\rightside(J)\in(\rightside(R(\cgT'_1)),\rightside(R(\cgT_m)))$.
This way, $J$ overlaps the intervals $R(\cgT'_1,\ldots,\cgT'_{\smash{m'}})$ and no other intervals, and thus merges $\cgT'_1,\ldots,\cgT'_{\smash{m'}}$ into one tree $\cgT_{m+1}$ with $R(\cgT_{m+1})=J$.
Again, the family $\cgJ\cup\cgJ'\cup\{J\}$ with trees $\cgT_1,\ldots,\cgT_{m+1}$ satisfies \ref{item:strategy-i} and \ref{item:strategy-ii}.
Moreover, Algorithm has to color $J$ with a color different from those used on $R(\cgT'_1),\ldots,R(\cgT'_{\smash{m'}})$.
Hence it has been forced to use at least $k$ colors on $R(\cgT_1),\ldots,R(\cgT_{m+1})$, so \ref{item:strategy-iii} is also satisfied.
See Figure \ref{fig:overlap-strategy} for an illustration.

In the final step, after playing the strategy for $k$ claimed above, we present one more interval $J$ such that $\leftside(J)\in(x,\rightside(R(\cgT_m)))$ and $\rightside(J)>\rightside(R(\cgT_1))$, where $x$ is the maximum of $\leftside(R(\cgT_m))$ and the right endpoints of all intervals in $\cgJ\setminus\{R(\cgT_1),\ldots,R(\cgT_m)\}$.
It follows that a $(k+1)$st color must be used on $J$.
This way, we have forced the use of more than $k$ colors on a family of at most $2^k$ intervals in the restricted overlap game.
\end{proof}

\begin{proof}[Proof of Lemma \ref{lem:construction} (sketch)]
According to what has been told in Section \ref{sec:overview}, any strategy of Presenter in the overlap game gives rise to an overlap game graph encoding this strategy.
Since the overlap graph of the family $\cgJ$ presented by the strategy described in the proof of Proposition \ref{prop:overlap-strategy} is triangle-free, the overlap game graph $G$ obtained from this strategy is triangle-free as well.
Since Algorithm is forced to use at least $k$ colors on $\cgJ$, the graph $G$ has chromatic number at least $k$ (actually, its chromatic number is exactly $k$).
One can calculate that the number of vertices of $G$ is less than $\frac{3}{2}\cdot\smash[t]{2^{2^{k-2}}}$.
See \cite{PKK+13} for more details.
\end{proof}

To show that any forest with fewer than $2^k$ vertices can be properly colored on-line using at most $k$ colors, we use the algorithm called \emph{First-fit}.
It colors each vertex with the least positive integer that has not been used on any neighbor presented before.

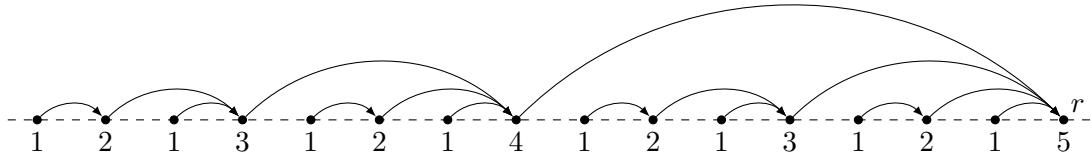
\begin{figure}[t]
\centering
\begin{tikzpicture}[scale=0.9,>=latex,label distance=2pt]
  \draw[dashed] (-0.43,0)--(15.43,0);
  \tikzstyle{every node}=[circle,minimum size=3pt,inner sep=0pt,draw,fill];
  \node[label=below:$1$] (a) at (0,0) {};
  \node[label=below:$2$] (b) at (1,0) {};
  \node[label=below:$1$] (c) at (2,0) {};
  \node[label=below:$3$] (d) at (3,0) {};
  \node[label=below:$1$] (e) at (4,0) {};
  \node[label=below:$2$] (f) at (5,0) {};
  \node[label=below:$1$] (g) at (6,0) {};
  \node[label=below:$4$] (h) at (7,0) {};
  \node[label=below:$1$] (i) at (8,0) {};
  \node[label=below:$2$] (j) at (9,0) {};
  \node[label=below:$1$] (k) at (10,0) {};
  \node[label=below:$3$] (l) at (11,0) {};
  \node[label=below:$1$] (m) at (12,0) {};
  \node[label=below:$2$] (n) at (13,0) {};
  \node[label=below:$1$] (o) at (14,0) {};
  \node[label=below:$5$,label=above right:$r$] (p) at (15,0) {};
  \path (a) edge[->,bend left=45] (b);
  \path (b) edge[->,bend left=45] (d);
  \path (c) edge[->,bend left=45] (d);
  \path (d) edge[->,bend left=45] (h);
  \path (e) edge[->,bend left=45] (f);
  \path (f) edge[->,bend left=45] (h);
  \path (g) edge[->,bend left=45] (h);
  \path (h) edge[->,bend left=45] (p);
  \path (i) edge[->,bend left=45] (j);
  \path (j) edge[->,bend left=45] (l);
  \path (k) edge[->,bend left=45] (l);
  \path (l) edge[->,bend left=45] (p);
  \path (m) edge[->,bend left=45] (n);
  \path (n) edge[->,bend left=45] (p);
  \path (o) edge[->,bend left=45] (p);
\end{tikzpicture}
\caption{Example of a tree $T$ with $f(r)=5$}
\label{fig:forcing-tree}
\end{figure}

\begin{proposition}[folklore]
\label{prop:first-fit}
First-fit uses at most\/ $k$ colors on any forest with fewer than\/ $2^k$ vertices presented on-line in any order.
\end{proposition}

\begin{proof}
Let $F$ be a forest with fewer than $2^k$ vertices presented on-line.
For $u,v\in V(F)$, let $u\edgeto v$ denote that $uv\in E(F)$ and $u$ has been presented before $v$.
Let $f(u)$ denote the color chosen by First-fit for each vertex $u$.
For every $v\in V(F)$, the colors $1,\ldots,f(v)-1$ have been chosen for some vertices $u$ with $u\edgeto v$.
Let $r$ be a vertex with maximum color, and let $T$ be a minimal subforest of $F$ that satisfies the following conditions (see Figure \ref{fig:forcing-tree}):
\begin{itemize}
\item $T$ contains the vertex $r$,
\item for any $v\in V(T)$, there are $u_1,\ldots,u_{f(v)-1}\in V(T)$ such that $u_j\edgeto v$ and $f(u_j)=j$ for $1\leq j\leq f(v)-1$.
\end{itemize}
It follows that there is a directed path in $T$ from every vertex to $r$, so $T$ is actually a tree.
Moreover, for every $u\in V(T)\setminus\{r\}$, there is a unique $v\in V(T)$ such that $u\edgeto v$.
Indeed, if there were two such vertices $v_1$ and $v_2$, then there would be two different paths in $T$ from $u$ to $r$, which cannot exist in a tree.
Let $n_j$ denote the number of vertices of $T$ with color $j$.
It follows that $n_{f(r)}=1$ and $n_j=n_{j+1}+\cdots+n_{f(r)}$ for $1\leq j\leq f(r)-1$.
Easy calculation yields $n_1+\cdots+n_{f(r)}=2^{f(r)-1}$.
Since $T$ has fewer than $2^k$ vertices, we have $f(r)\leq k$, which means that First-fit has used at most $k$ colors on $F$.
\end{proof}

\section{Coloring triangle-free overlap game graphs}
\label{sec:coloring}

For the purpose of this entire section, let $G$ be an $n$-vertex triangle-free overlap game graph with meta-forest $M$ and representation $\mu$.
Our goal is to prove that $G$ has chromatic number $O(\log\log n)$.
Since different components of $M$ are not connected by edges of $G$, they can be colored independently using the same set of colors.
Thus it is enough to consider each component of $M$ separately, and therefore we can assume without loss of generality that $M$ is a single tree.

For any set $S$ of vertices lying on a common root-to-leaf path of $M$, the family of intervals $\{\mu(v)\colon v\in S\}$ excludes the two configurations of intervals presented in Figures \ref{fig:restricted-clean}, by the condition \ref{def:forbidden-structure} on $G$, and \ref{fig:restricted-triangle}, by the assumption that $G$ is triangle-free, but it may contain the configuration in Figure \ref{fig:restricted-secondary}.
Our first goal is to reduce the problem to the case that all three configurations in Figure \ref{fig:restricted} are excluded.

We classify each vertex of $G$ as either \emph{primary} or \emph{secondary} according to the following inductive rule.
The root of $M$ is primary.
Now, let $x$ be a vertex other than the root of $M$, and suppose that all vertices $v$ with $v\prec x$ have been already classified.
If there are primary vertices $u$ and $v$ with $u\prec v\prec x$ such that $\mu(u)$, $\mu(v)$ and $\mu(x)$ form the configuration in Figure \ref{fig:restricted-secondary}, that is, $\mu(v)\supset\mu(x)$ and $\mu(u)$ overlaps both $\mu(v)$ and $\mu(x)$, then $x$ is secondary.
Otherwise, $x$ is primary.
It clearly follows that for any set $P$ of primary vertices lying on a common root-to-leaf path of $M$, the family of intervals $\{\mu(v)\colon v\in P\}$ excludes all three configurations in Figure \ref{fig:restricted}, that is, it is restricted.

For every primary vertex $v$, let $S(v)$ be the set consisting of $v$ and all vertices $x$ with $v\prec x$ for which there is a primary vertex $u$ with $u\prec v$ such that $\mu(u)$, $\mu(v)$ and $\mu(x)$ form the configuration in Figure \ref{fig:restricted-secondary}.
Hence every primary or secondary vertex belongs to some $S(v)$.

\begin{lemma}
\label{lem:S}
For every primary vertex\/ $v$, the set\/ $S(v)$ consists of those and only those vertices\/ $x$ for which\/ $v$ is the last primary vertex on the path from the root to\/ $x$ in\/ $M$.
Moreover, every\/ $S(v)$ is an independent set in\/ $G$.
\end{lemma}

\begin{proof}
Let $v$ be a primary vertex, and let $y\in S(v)$.
It follows that there is a primary vertex $u$ with $u\prec v\prec y$ such that $\mu(u)$, $\mu(v)$ and $\mu(y)$ form the configuration in Figure \ref{fig:restricted-secondary}.
Let $x$ be any vertex with $v\prec x\prec y$.
We claim that $\mu(v)\supset\mu(x)\supset\mu(y)$ (see Figure \ref{fig:primary}).
Indeed,
\begin{itemize}
\item if $\rightside(\mu(x))<\rightside(\mu(u))$, then $\mu(u)$, $\mu(v)$ and $\mu(x)$ form the configuration in Figure \ref{fig:restricted-clean}, which is excluded by the condition \ref{def:forbidden-structure} on $G$;
\item if $\rightside(\mu(u))<\rightside(\mu(x))<\rightside(\mu(y))$, then $\mu(u)$, $\mu(x)$ and $\mu(y)$ form the configuration in Figure \ref{fig:restricted-triangle}, which contradicts the assumption that $G$ is triangle-free;
\item if $\rightside(\mu(x))>\rightside(\mu(v))$, then $\mu(u)$, $\mu(v)$ and $\mu(x)$ form the configuration in Figure \ref{fig:restricted-triangle}, which again contradicts the assumption that $G$ is triangle-free;
\end{itemize}
hence $\rightside(\mu(y))<\rightside(\mu(x))<\rightside(\mu(v))$, that is, $\mu(v)\supset\mu(x)\supset\mu(y)$.
It follows that $\mu(u)$, $\mu(v)$ and $\mu(x)$ also form the configuration in Figure \ref{fig:restricted-secondary}, so $x$ is secondary and $x\in S(v)$.
This yields the first statement of the lemma.
It also follows that $\mu(x)\supset\mu(y)$ whenever $x,y\in S(v)$ and $x\prec y$, which proves that $S(v)$ is an independent set in $G$.
\end{proof}

\begin{figure}[t]
\centering
\begin{minipage}[b]{0.45\linewidth}
\centering
\begin{tikzpicture}[scale=0.3,yscale=1.8]
  \draw[|-|] (-8,0)--(2,0);
  \draw[|-|] (-6,0.8)--(8,0.8);
  \draw[|-|,dashed] (-4,1.5)--(10,1.5);
  \draw[|-|] (-4,2)--(6,2);
  \draw[|-|,dashed] (-4,2.5)--(3.5,2.5);
  \draw[|-|,dashed] (-4,3)--(0.5,3);
  \draw[|-|,dashed] (-4,3.5)--(-2,3.5);
  \draw[|-|] (-1,4.2)--(5,4.2);
  \useasboundingbox (current bounding box);
  \node[inner sep=2pt,fill=white] at (-3,0) {$u$};
  \node[inner sep=2pt,fill=white] at (1,0.8) {$v$};
  \node[inner sep=2pt,fill=white] at (1,2) {$x$};
  \node[inner sep=2pt,fill=white] at (2,4.2) {$y$};
\end{tikzpicture}
\caption{Illustration for the proof of Lemma \ref{lem:S}}
\label{fig:primary}
\end{minipage}\hspace{0.5cm}
\begin{minipage}[b]{0.45\linewidth}
\centering
\begin{tikzpicture}[scale=0.3,yscale=1.3]
  \draw[|-|] (-10,0)--(10,0);
  \draw[-|,dashed] (-2.5,1)--(-9,1);
  \draw[-|,dashed] (-2.5,1)--(4,1);
  \draw[|-|] (-8,2)--(8,2);
  \draw[-|,dashed] (-2.5,3)--(-7,3);
  \draw[-|,dashed] (-2.5,3)--(2,3);
  \draw[|-|] (-1,4)--(7,4);
  \draw[-|,dashed] (3,5)--(0,5);
  \draw[-|,dashed] (3,5)--(6,5);
  \draw[-|,dashed] (3,6)--(1,6);
  \draw[-|,dashed] (3,6)--(5,6);
  \useasboundingbox (current bounding box);
  \node[inner sep=2pt,fill=white] at (0,0) {$u_1$};
  \node[inner sep=2pt,fill=white] at (-2.5,1) {$x_1$};
  \node[inner sep=2pt,fill=white] at (0,2) {$u_2$};
  \node[inner sep=2pt,fill=white] at (-2.5,3) {$x_2$};
  \node[inner sep=2pt,fill=white] at (3,4) {$v$};
  \node[inner sep=2pt,fill=white] at (3,5) {$y_1$};
  \node[inner sep=2pt,fill=white] at (3,6) {$y_2$};
\end{tikzpicture}
\caption{Illustration for the proof of Lemma \ref{lem:secondary}}
\label{fig:secondary}
\end{minipage}
\end{figure}

The relation $\prec$ defines an orientation of the edges of $G$: we write $x\edgeto y$ if $xy\in E(G)$ and $x\prec y$.
We write $S(u)\edgeto S(v)$ if $u\prec v$ and there are $x\in S(u)$ and $y\in S(v)$ such that $x\edgeto y$.
We are going to show that any proper $k$-coloring of the primary vertices of $G$ can be transformed into a proper $2k$-coloring of the whole $G$.
This is done with the help of the following lemma.

\begin{lemma}
\label{lem:secondary}
Let\/ $I$ be an independent set of primary vertices in\/ $G$, and let\/ $v\in I$.
There is at most one vertex\/ $u\in I$ such that\/ $S(u)\edgeto S(v)$.
\end{lemma}

\begin{proof}
Suppose there are two vertices $u_1,u_2\in I$ such that $u_1\prec u_2\prec v$, $S(u_1)\edgeto S(v)$, and $S(u_2)\edgeto S(v)$.
It follows that there are vertices $x_1\in S(u_1)$, $x_2\in S(u_2)$, and $y_1,y_2\in S(v)$ with $x_1\edgeto y_1$ and $x_2\edgeto y_2$.
Hence, by Lemma \ref{lem:S}, we have $u_1\prec x_1\prec u_2\prec x_2\prec v\prec y_i$ for $i\in\{1,2\}$.
Since $\mu(x_1)\subset\mu(u_1)$, $\mu(x_2)\subset\mu(u_2)$, $\mu(y_1)\subset\mu(v)$, $\mu(y_2)\subset\mu(v)$, and $u_1$, $u_2$ and $v$ are independent in $G$, it follows that $\mu(u_1)$, $\mu(u_2)$ and $\mu(v)$ are nested (see Figure \ref{fig:secondary}).
Clearly, $\mu(x_1)$ and $\mu(x_2)$ overlap $\mu(v)$, and $\mu(x_1)$ overlaps $\mu(u_2)$.
If $\rightside(\mu(x_1))<\rightside(\mu(x_2))$, then $\mu(x_1)$, $\mu(x_2)$ and $\mu(v)$ form the configuration in Figure \ref{fig:restricted-triangle}, which contradicts the assumption that $G$ is triangle-free.
If $\rightside(\mu(x_1))>\rightside(\mu(x_2))$, then $\mu(x_1)$, $\mu(u_2)$ and $\mu(x_2)$ form the configuration in Figure \ref{fig:restricted-clean}, which is excluded by \ref{def:forbidden-structure}.
\end{proof}

We now show how to color the vertices of $G$ with $2k$ colors.
Let $I$ be a color class in a proper $k$-coloring of the primary vertices of $G$.
Consider all the sets $S(u)$ for $u\in I$.
By Lemma \ref{lem:S}, each of them is independent.
By Lemma \ref{lem:secondary}, the edges between the sets form a bipartite graph.
Therefore, we need two colors for the vertices in $\bigcup_{u\in I}S(u)$ and hence $2k$ colors for the whole~$G$.

It remains to show that the chromatic number of the subgraph of $G$ induced on the primary vertices is $O(\log\log n)$.
Clearly, this subgraph is itself an overlap game graph.
Therefore, from now on, we simply assume that all vertices of $G$ are primary.
This means that for every set $P$ of vertices lying on a common root-to-leaf path of $M$, the family of intervals $\{\mu(v)\colon v\in P\}$ is restricted.
The following is an immediate consequence of Lemma \ref{lem:restricted}.

\begin{lemma}
\label{lem:p-restricted}
Let\/ $P$ be a set of vertices lying on a common root-to-leaf path of\/ $M$.
\begin{enumeratea}
\item\label{lem:p-forest} For every\/ $u\in P$, there is at most one\/ $v\in P$ such that\/ $u\edgeto v$.
In particular, the graph\/ $G[P]$ is a forest.
\item\label{lem:p-interlace} There are no\/ $u_1,u_2,v_1,v_2\in P$ such that\/ $u_1\prec u_2\prec v_1\prec v_2$, $u_1\edgeto v_1$, and\/ $u_2\edgeto v_2$.
\end{enumeratea}
\end{lemma}

Now, the simplest idea would be to color the vertices by First-fit, processing them from left to right in the order $\prec$.
Then, by Lemma \ref{lem:p-restricted} \ref{lem:p-forest} and Proposition \ref{prop:first-fit}, the number of colors used would be logarithmic in the maximum length of a root-to-leaf path in $M$.
This is not enough when $M$ contains paths longer than $O(\log n)$.
To overcome this problem, we need to introduce the concept of heavy-light decomposition due to Sleator and Tarjan \cite{ST83}.

Let $T$ be a rooted tree.
We call an edge $uv$ of $T$, where $v$ is a child of $u$, \emph{heavy} if the subtree of $T$ rooted at $v$ contains more than half of the vertices of the subtree of $T$ rooted at $u$.
Otherwise, we call the edge $uv$ \emph{light}.
The resulting partition of the edges of $T$ into heavy and light edges is called the \emph{heavy-light decomposition} of $T$.
Since every vertex $u$ of $T$ has a heavy edge to at most one of its children, the heavy edges induce in $T$ a collection of paths, called \emph{heavy paths}.
The heavy-light decomposition has the following crucial property, easily proved by induction.

\begin{lemma}[Sleator, Tarjan \cite{ST83}]
\label{lem:heavy-light}
If there is a path in\/ $T$ starting at the root and containing at least\/ $k-1$ light edges, then\/ $T$ has at least\/ $2^k-1$ vertices.
\end{lemma}

Fix a heavy-light decomposition of $M$.
Form an auxiliary graph $G'$ by removing from $G$ the edges connecting pairs of vertices in different heavy paths.
By Lemma \ref{lem:p-restricted} \ref{lem:p-forest}, the vertices on each heavy path induce a forest in $G$.
Hence $G'$ is a forest and can be properly colored with two colors.
Let $C_1$ and $C_2$ be the color classes in a proper two-coloring of $G'$.
Fix $i\in\{1,2\}$.
It follows that for any $x,y\in C_i$, if $x\edgeto y$, then $x$ and $y$ lie on different heavy paths.
Color the induced subgraph $G[C_i]$ of $G$ by First-fit, processing the vertices from left to right in the order $\prec$.
That is, the color assigned to a vertex $v\in C_i$ is the least positive integer not assigned to any $u\in C_i$ with $u\edgeto v$.

\begin{lemma}
\label{lem:ff-heavy}
If First-fit assigns a color\/ $k$ to some vertex\/ $r\in C_i$, then the path\/ $M_r$ in\/ $M$ from the root to\/ $r$ contains at least\/ $2^{k-2}-1$ light edges.
\end{lemma}

\begin{proof}
Let $f(u)$ denote the color chosen by First-fit for each vertex $u\in C_i$.
Since $G[C_i]$ is a forest, we can find, as in the proof of Proposition \ref{prop:first-fit}, a minimal subtree $T$ of $G[C_i]$ that satisfies the following conditions (see Figure \ref{fig:forcing-tree}):
\begin{itemize}
\item $T$ contains the vertex $r$,
\item for any $v\in V(T)$, there are $u_1,\ldots,u_{f(v)-1}\in V(T)$ such that $u_j\edgeto v$ and $f(u_j)=j$ for $1\leq j\leq f(v)-1$.
\end{itemize}
Clearly, the entire $T$ lies on the path $M_r$.
Moreover, as in the proof of Proposition \ref{prop:first-fit}, $T$ contains exactly $2^{k-1}$ vertices, $2^{k-2}$ of which have color $1$ and $2^{k-2}$ have color greater than~$1$.
Let $v$ be a vertex of $T$ with color greater than $1$, and let $u\in V(T)$ be the vertex directly preceding $v$ in the order $\prec$ on $V(T)$.
We claim that $u$ is a child of $v$ in $T$.
Suppose it is not.
Since $u$ is not the $\prec$-greatest vertex of $T$, we have $u\neq r$, and therefore $u$ has a parent $p$ in $T$.
Since $v$ has color greater than $1$, it has a child $c$ in $T$.
It follows that $c\prec u\prec v\prec p$, as $u$ and $v$ are consecutive in the order $\prec$ on $V(T)$.
But we also have $c\edgeto v$ and $u\edgeto p$, which contradicts Lemma \ref{lem:p-restricted} \ref{lem:p-interlace}.
We have shown that there is an edge between $u$ and $v$ in $T$, so they must lie on different heavy paths.
Consequently, each vertex of $T$ with color greater than $1$ belongs to a different heavy path.
This shows that the path $M_r$ contains at least $2^{k-2}-1$ light edges.
\end{proof}

\begin{proof}[Proof of Lemma \ref{lem:game-graph-coloring}]
As noted before, we may assume that $G$ consists only of primary vertices and $M$ is a tree.
Let $k$ be the maximum color used by First-fit on a vertex $r\in C_i$.
By Lemma \ref{lem:ff-heavy}, the path $M_r$ in $M$ from the root to $r$ contains at least $2^{k-2}-1$ light edges.
This implies, by Lemma \ref{lem:heavy-light}, that $n\geq 2^{2^{k-2}}-1$.
Therefore, First-fit uses at most $O(\log\log n)$ colors on $C_i$.
We color $C_1$ and $C_2$ by First-fit using two separate sets of colors, obtaining a proper coloring of $G$ with $O(\log\log n)$ colors.
\end{proof}

\section{Clean directed families of frames = overlap game graphs}
\label{sec:relation}

\begin{proof}[Proof of Lemma \ref{lem:relation}]
First, we show that every overlap game graph is isomorphic to the intersection graph of a clean directed family of frames.
Let $G$ be an overlap game graph with meta-forest $M$ and representation $\mu$.
We define frames $F(v)$ for all vertices $v$ of $G$ so that
\begin{itemize}
\item $\leftside(F(v))=\leftside(\mu(v))$ and $\rightside(F(v))=\rightside(\mu(v))$,
\item $\bottomside(F(v))<\bottomside(F(v_1))<\topside(F(v_1))<\cdots<\bottomside(F(v_k))<\topside(F(v_k))<\topside(F(v))$, where $v_1,\ldots,v_k$ are the children of $v$ in $M$ in any order.
\end{itemize}
See Figure \ref{fig:correspondence} for an illustration.
The numbers $\bottomside(F(v))$ and $\topside(F(v))$ can be computed by performing a depth-first search over the forest $M$ and recording, for each vertex $v$, the times at which the subtree of $M$ rooted at $v$ is entered and left, respectively.
Clearly, if $u$ and $v$ are two vertices unrelated in $M$, then $\topside(F(u))<\bottomside(F(v))$ or $\topside(F(v))<\bottomside(F(u))$, so $F(u)$ and $F(v)$ do not intersect.
If $u\prec v$, then we have $\bottomside(F(u))<\bottomside(F(v))<\topside(F(v))<\topside(F(u))$, and therefore $F(u)$ and $F(v)$ intersect if and only if $\mu(u)$ and $\mu(v)$ overlap.
Moreover, when $F(u)$ and $F(v)$ intersect, this intersection is rightward-directed.
If there are three vertices $u$, $v$, $w$ such that $F(u)$ and $F(v)$ intersect and $F(w)$ is enclosed in both $F(u)$ and $F(v)$, then $u\prec v\prec w$, $\mu(u)$ and $\mu(v)$ overlap, and $\mu(w)\subset\mu(u)\cap\mu(v)$, which is excluded by the condition \ref{def:forbidden-structure} of the definition of an overlap game graph.
Consequently, $\{F(v)\colon v\in V(G)\}$ is a clean rightward-directed family of frames, and its intersection graph is isomorphic to $G$.

\begin{figure}[t]
\centering
\begin{tikzpicture}[scale=1.4,yscale=1.2,shorten <=-0.2pt,shorten >=-0.2pt,label distance=2pt]
  \draw[dotted,thick] (0,0) -- (1,0) -- ++ (30:1cm) -- ++ (1,0);
  \draw[dotted,thick] (0,0) -- (1,0) -- ++ (-30:1cm);
  \draw[|-|] (0,-1.3)--(0,1.3);
  \draw[|-|] (1,-1.1)--(1,1.1);
  \draw[|-|] (1.866,0.1)--(1.866,0.9);
  \draw[|-|] (2.866,0.3)--(2.866,0.7);
  \draw[|-|] (1.866,-0.8)--(1.866,-0.2);
  \tikzstyle{every node}=[circle,minimum size=3pt,inner sep=0pt,draw,fill]
  \node[label=left:$a$] (a) at (0,0) {};
  \node[label=above left:$b$] (b) at (1,0) {};
  \node[label=above left:$c$] (c) at (1.866,0.5) {};
  \node[label=right:$d$] (d) at (2.866,0.5) {};
  \node[label=right:$e$] (e) at (1.866,-0.5) {};
  \path (c) edge[bend right=30] (b);
  \path (d) edge[bend right=30] (c) edge[bend right=42] (a);
  \path (e) edge[bend left=30] (b) edge[bend left=30] (a);
  \draw[dashed] (0,-1.3)--(5,-1.3);
  \draw[dashed] (0,1.3)--(5,1.3);
  \draw[dashed] (1,-1.1)--(5.5,-1.1);
  \draw[dashed] (1,1.1)--(5.5,1.1);
  \draw[dashed] (1.866,0.1)--(6.5,0.1);
  \draw[dashed] (1.866,0.9)--(6.5,0.9);
  \draw[dashed] (2.866,0.3)--(7.5,0.3);
  \draw[dashed] (2.866,0.7)--(7.5,0.7);
  \draw[dashed] (1.866,-0.8)--(6.5,-0.8);
  \draw[dashed] (1.866,-0.2)--(6.5,-0.2);
  \draw[|-|] (5,3)--(8.5,3);
  \draw[|-|] (5.5,2.75)--(6.8,2.75);
  \draw[|-|] (6.2,2.5)--(9.1,2.5);
  \draw[|-|] (6.5,2.25)--(7.8,2.25);
  \draw[|-|] (7.5,2)--(8.8,2);
  \draw[dashed] (5,3)--(5,1.3);
  \draw[dashed] (8.5,3)--(8.5,1.3);
  \draw[dashed] (5.5,2.75)--(5.5,1.1);
  \draw[dashed] (6.8,2.75)--(6.8,1.1);
  \draw[dashed] (6.2,2.5)--(6.2,-0.2);
  \draw[dashed] (9.1,2.5)--(9.1,-0.2);
  \draw[dashed] (6.5,2.25)--(6.5,0.9);
  \draw[dashed] (7.8,2.25)--(7.8,0.9);
  \draw[dashed] (7.5,2)--(7.5,0.7);
  \draw[dashed] (8.8,2)--(8.8,0.7);
  \tikzstyle{every node}=[inner sep=2pt,fill=white]
  \node at (6.75,3) {$\mu(a)$};
  \node at (6.15,2.75) {$\mu(b)$};
  \node at (7.65,2.5) {$\mu(e)$};
  \node at (7.15,2.25) {$\mu(c)$};
  \node at (8.15,2) {$\mu(d)$};
  \draw (5,-1.3) rectangle (8.5,1.3);
  \draw (5.5,-1.1) rectangle (6.8,1.1);
  \draw (6.5,0.1) rectangle (7.8,0.9);
  \draw (7.5,0.3) rectangle (8.8,0.7);
  \draw (6.2,-0.8) rectangle (9.1,-0.2);
\end{tikzpicture}
\caption{Correspondence between overlap game graphs and intersection graphs of clean directed families of frames}
\label{fig:correspondence}
\end{figure}
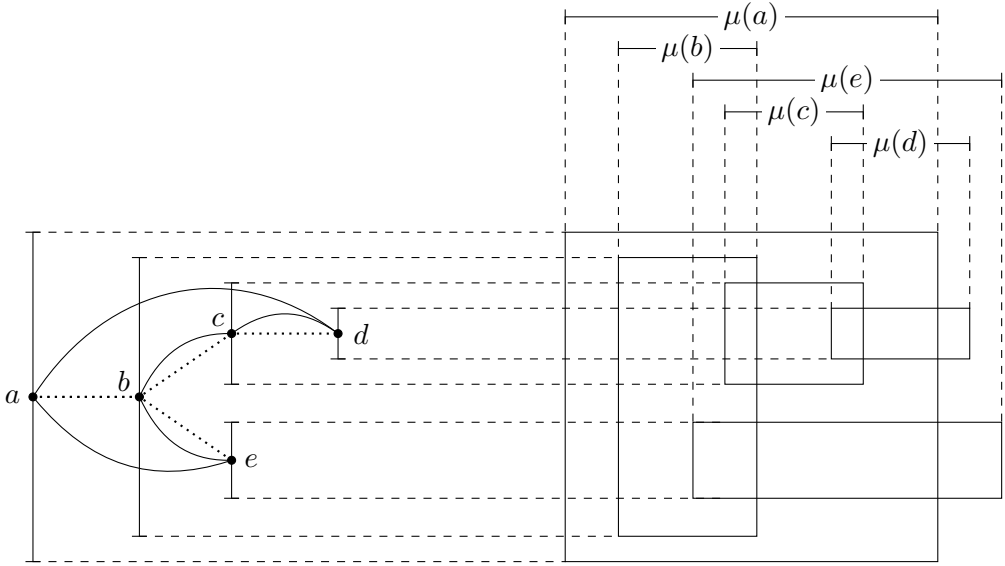

It remains to show that the intersection graph of every clean directed family of frames is an overlap game graph.
Let $\cgF$ be a directed family of frames.
We can assume without loss of generality that $\cgF$ is rightward-directed.
Define a mapping $\mu\colon\cgF\to\cgI$ so that $\leftside(\mu(F))=\leftside(F)$ and $\rightside(\mu(F))=\rightside(F)$, that is, $\mu(F)$ is the interval obtained by projecting $F$ on the $x$-axis.

For $F\in\cgF$, let $\cgL(F)$ be the subfamily of $\cgF$ consisting of those $F'$ for which $\leftside(F')<\leftside(F)<\rightside(F')$ and $\bottomside(F')<\bottomside(F)<\topside(F)<\topside(F')$.
We define a rooted forest $M$ on $\cgF$ as follows.
If $\cgL(F)=\emptyset$, then $F$ is a root of $M$.
Otherwise, the parent of $F$ in $M$ is the member $F'$ of $\cgL(F)$ with greatest $\leftside(F')$.
We show that the intersection graph of $\cgF$ is an overlap game graph with meta-forest $M$ and representation $\mu$.
To this end, we argue that the conditions \ref{def:increasing}--\ref{def:forbidden-structure} of the definition of an overlap game graph are satisfied by the intersection graph of $\cgF$.

It follows directly from the definition of parents that $F_1\prec F_2$ implies $\leftside(F_1)<\leftside(F_2)$ and $\bottomside(F_1)<\bottomside(F_2)<\topside(F_2)<\topside(F_1)$.
This already shows \ref{def:increasing} and the right-to-left implication in \ref{def:game-graph}.
Suppose there are $F_1,F_2,F_3\in\cgF$ such that $F_1\prec F_2\prec F_3$, $\mu(F_1)$ and $\mu(F_2)$ overlap, and $\mu(F_3)\subset\mu(F_1)\cap\mu(F_2)$.
We have $\leftside(F_1)<\leftside(F_2)<\leftside(F_3)<\rightside(F_3)<\rightside(F_1)<\rightside(F_2)$ and $\bottomside(F_1)<\bottomside(F_2)<\bottomside(F_3)<\topside(F_3)<\topside(F_2)<\topside(F_1)$, a configuration that is forbidden in a clean family of frames.
This contradiction shows \ref{def:forbidden-structure}.
Now, let $F_1$ and $F_2$ be two intersecting members of $\cgF$.
By the assumption that $\cgF$ is rightward-directed, $\mu(F_1)$ and $\mu(F_2)$ overlap, and we have $F_1\in\cgL(F_2)$ or $F_2\in\cgL(F_1)$.
Therefore, in order to prove the left-to-right implication in \ref{def:game-graph}, it remains to show that $F_1\in\cgL(F_2)$ implies $F_1\prec F_2$.
To this end, we use induction on the increasing order of $\leftside(F_2)$.
There is nothing to prove when $F_1$ is the parent of $F_2$.
Hence assume that $F_1\in\cgL(F_2)$, but $F_1$ is not the parent of $F_2$.
Let $F_2'$ be the parent of $F_2$.
We have $\leftside(F_1)<\leftside(F_2')<\leftside(F_2)<\rightside(F_1)$ and, since $\cgF$ is rightward-directed, $\bottomside(F_1)<\bottomside(F_2')<\topside(F_2')<\topside(F_1)$.
Thus $F_1\in\cgL(F_2')$.
This and the induction hypothesis yield $F_1\prec F_2'$ and thus $F_1\prec F_2$.
\end{proof}

\section{Reduction to clean directed families of frames}
\label{sec:geometry}

The goal of this section is to prove Lemmas \ref{lem:frame-geometry} and \ref{lem:frame-clean}.
To make the description simpler, we are going to partition a family of frames with bounded clique number into a bounded number of subfamilies with the property that the connected components of each subfamily are directed, but the directions of different components may be different.
Let $\xi(\cgF)$ denote the minimum size of such a partition of a family of frames $\cgF$.
It is enough to bound $\xi(\cgF)$, because we can gather the connected components of each partition class that have the same direction, thus turning each partition class into at most four directed families.
This way we will obtain the bound of $4\xi(\cgF)$ on the size of a partition of $\cgF$ into directed families.

Two families of frames $\cgF_1$ and $\cgF_2$ are \emph{mutually disjoint} if there are no two intersecting frames $F_1\in\cgF_1$ and $F_2\in\cgF_2$.
The following properties of $\xi$ (analogous to the properties of chromatic number) are straightforward:
\begin{enumeratei}
\item $\xi(\cgF_1\cup\cdots\cup\cgF_m)\leq\xi(\cgF_1)+\cdots+\xi(\cgF_m)$;
\item if $\cgF_i$ and $\cgF_j$ are mutually disjoint for $i\neq j$, then $\xi(\cgF_1\cup\cdots\cup\cgF_m)=\max\{\xi(\cgF_1),\ldots,\xi(\cgF_m)\}$.
\end{enumeratei}
We will use them implicitly in the rest of this section.
For convenience, whenever we consider a partition of a family into a number of subfamilies, we allow these subfamilies to be empty.

Recall that two frames are said to cross if both vertical sides of one intersect both horizontal sides of the other.
A family of frames $\cgF$ is \emph{non-crossing} if it contains no two crossing frames.

\begin{lemma}[implicit in \cite{AG60}]
\label{lem:cross}
Every family of frames\/ $\cgF$ with\/ $\omega(\cgF)\leq k$ can be partitioned into\/ $k$ non-crossing subfamilies.
\end{lemma}

\begin{proof}
Define a partial order $<$ on $\cgF$ so that $F_1<F_2$ whenever both vertical sides of $F_1$ intersect both horizontal sides of $F_2$.
The graph $G^+$ of crossing pairs of frames in $\cgF$ is the comparability graph of $<$, so it is perfect.
Hence $\chi(G^+)=\omega(G^+)\leq\omega(\cgF)\leq k$, and the lemma follows.
\end{proof}

A frame $E$ is \emph{external} to a family of frames $\cgF$ if $E\not\subset\bigcup_{F\in\cgF}\rect(F)$.
A subfamily $\cgL$ of a family of frames $\cgF$ is a \emph{layer} of $\cgF$ if $|\cgL|=1$ or every frame in $\cgL$ intersects some frame in $\cgF\setminus\cgL$ external to $\cgL$.
A family $\cgL\subset\cgF$ is an \emph{$m$-fold layer} of $\cgF$ if there is a chain of families $\cgL=\cgL_m\subset\cgL_{m-1}\subset\cdots\subset\cgL_0=\cgF$ such that $\cgL_{i+1}$ is a layer of $\cgL_i$ for $0\leq i\leq m-1$.
Note that every subfamily of a layer (an $m$-fold layer) of $\cgF$ is also a layer (an $m$-fold layer) of $\cgF$.

\begin{lemma}
\label{lem:bfs}
Every family of frames\/ $\cgF$ has a partition\/ $\scrP$ into layers.
Moreover, $\scrP$ can be partitioned into two subclasses each consisting of mutually disjoint layers.
\end{lemma}

\begin{proof}
When $\cgF$ is not connected, the layers and the bipartition can be constructed for each component separately.
Thus assume, without loss of generality, that $\cgF$ is connected.
Choose any frame $R\in\cgF$ that is external to $\cgF\setminus\{R\}$.
For $j\in\setN$, let $\cgL_j$ consist of those frames in $\cgF$ whose distance to $R$ in the intersection graph of $\cgF$ is $j$.
Let $d$ be greatest such that $\cgL_d\neq\emptyset$.
We claim that $\scrP=\{\cgL_0,\ldots,\cgL_d\}$ is a partition of $\cgF$ into layers.

We have $|\cgL_0|=1$, so $\cgL_0$ is a layer.
To show that $\cgL_j$ with $j\geq 1$ is a layer, we find, for each $F\in\cgL_j$, a frame in $\cgF\setminus\cgL_j$ external to $\cgL_j$ and intersecting $F$.
Let $F_0F_1\ldots F_j$ be a path from $F_0=R$ to $F_j=F$ of length $j$.
Let $i$ be greatest such that $F_i$ is external to $\cgL_j$ (such $i$ exists, as $R$ is external to $\cgL_j$).
Since $F_i$ intersects $F_{i+1}$, which is not external to $\cgL_j$, $F_i$ must intersect some frame $F'\in\cgL_j$.
This witnesses a path from $R$ to $F'$ of length $i+1$, so $i+1=j$.
Therefore, $F_i$ is the requested frame external to $\cgL_j$ and intersecting $F$.

Clearly, any two intersecting frames belong to one layer or two layers with consecutive indices.
Hence the families $\{\cgL_j\colon j$ is odd$\}$ and $\{\cgL_j\colon j$ is even$\}$ consist of mutually disjoint layers.
\end{proof}

\begin{corollary}
\label{cor:bfs}
Every family of frames\/ $\cgF$ has a partition\/ $\scrP$ into\/ $m$-fold layers.
Moreover, $\scrP$ can be partitioned into\/ $2^m$ subclasses each consisting of mutually disjoint\/ $m$-fold layers.
\end{corollary}

\begin{proof}
We proceed by induction on $m$.
For $m=1$, this is exactly the statement of Lemma \ref{lem:bfs}.
For $m\geq 2$, we apply induction hypothesis to construct a partition $\scrP'$ of $\cgF$ into $(m-1)$-fold layers and a partition of $\scrP'$ into $2^{m-1}$ subclasses $\scrP_1,\ldots,\scrP_{2^{m-1}}$ each consisting of mutually disjoint layers.
We apply Lemma \ref{lem:bfs} to construct a partition $\scrP^\cgL$ of each $(m-1)$-fold layer $\cgL\in\scrP'$ into layers, which are now $m$-fold layers of $\cgF$, and a partition of $\scrP^\cgL$ into two subclasses $\scrP^\cgL_1$ and $\scrP^\cgL_2$ each consisting of mutually disjoint layers.
We set $\scrP=\bigcup_{\cgL\in\scrP}\scrP^\cgL$.
The desired partition of $\scrP$ is formed by the families $\bigcup_{\cgL\in\scrP_i}\scrP^\cgL_1$ and $\bigcup_{\cgL\in\scrP_i}\scrP^\cgL_2$ for $1\leq i\leq 2^{m-1}$.
\end{proof}

\begin{theorem}[Asplund, Gr\"unbaum \cite{AG60}]
\label{thm:rect}
Let\/ $\cgR$ be a family of axis-parallel rectangles with\/ $\omega(\cgR)\leq k$.
It follows that\/ $\chi(\cgR)\leq\alpha_k$, where\/ $\alpha_k$ depends only on\/ $k$.
\end{theorem}

The precise bound of Theorem \ref{thm:rect} in \cite{AG60} is $4k^2-3k$.
This was improved to $3k^2-2k-1$ by Hendler \cite{Hen98}.
However, we are going to apply Theorem \ref{thm:rect} only to non-crossing families $\cgR$, for which the bound in \cite{AG60} is $4k-3$.

A curve is \emph{$x$-semimonotone} if its intersection with every vertical line is connected or empty.

\begin{theorem}[Suk \cite{Suk14}]
\label{thm:suk}
Let\/ $V$ be a vertical line.
Let\/ $\cgC$ be a family of\/ $x$-semimonotone curves contained in the right half-plane delimited by\/ $V$ such that\/ $|V\cap C|=1$ for any\/ $C\in\cgC$, $|C_1\cap C_2|\leq 1$ for any distinct\/ $C_1,C_2\in\cgC$, and\/ $\omega(\cgC)\leq k$.
It follows that\/ $\chi(\cgC)\leq\beta_k$, where\/ $\beta_k$ depends only on\/ $k$.
\end{theorem}

The actual theorem in \cite{Suk14} concerns families of $x$-monotone curves, that is, curves whose intersection with every vertical line is a single point or empty.
The generalization to $x$-semimonotone curves above comes from the fact that we can always perturb a family of $x$-semimonotone curves to make them $x$-monotone without changing their intersection graph.
Theorem \ref{thm:suk} for $k=2$ but without the $x$-semimonotonicity assumption was proved by McGuinness \cite{McG00}.
Laso\'n et al.\ \cite{LMPW14} generalized Theorem \ref{thm:suk} removing the $x$-semimonotonicity assumption for any $k$.

For the rest of this section, assume that $\alpha_k$ and $\beta_k$ are as in the statements of Theorems \ref{thm:rect} and \ref{thm:suk}.
Define constants $\gamma_k$ and $\delta_{k,m}$ for $1\leq m\leq k$ by induction on $k$ and $m$, as follows:
\begin{equation*}
\gamma_1=1,\qquad\delta_{k,1}=0,\qquad\delta_{k,m}=4(m-1)(\beta_k^4+2\gamma_{k-1}+\delta_{k,m-1}),\qquad\gamma_k=2^k\alpha_k\delta_{k,k}.
\end{equation*}
We are going to prove the bound $\xi(\cgF)\leq k\gamma_k$ for families of frames $\cgF$ with $\omega(\cgF)\leq k$.

\begin{lemma}
\label{lem:geometry-inner}
Let\/ $\cgF$ be a non-crossing family of frames with\/ $\omega(\cgF)\leq k$, let $1\leq m\leq k$, let $B_1,\ldots,B_{k-m+1}$ be pairwise intersecting frames, and let\/ $\cgL\subset\cgF$ be a family of frames enclosed in each of\/ $B_1,\ldots,B_{k-m+1}$ such that\/ $\cgL\cup\{B_1,\ldots,B_{k-m+1}\}$ is an\/ $m$-fold layer of\/ $\cgF$.
Suppose that\/ $\xi(\cgL')\leq\gamma_{k-1}$ holds for any subfamily\/ $\cgL'\subset\cgL$ with\/ $\omega(\cgL')\leq k-1$.
It follows that\/ $\xi(\cgL)\leq\delta_{k,m}$.
\end{lemma}

\begin{proof}
The proof goes by induction on $m$.
If $m=1$, then the assumption that $\cgL\cup\{B_1,\ldots,B_k\}$ is a layer of $\cgF$ yields $\cgL=\emptyset$.
Indeed, if there was some $F\in\cgL$, then there would be a frame $E\in\cgF$ external to $\cgL\cup\{B_1,\ldots,B_k\}$ and intersecting $F$.
Hence $\{B_1,\ldots,B_k,E\}$ would be a $(k+1)$-clique in $\cgF$.
We conclude that $\xi(\cgL)=0=\delta_{k,1}$.

Now, assume that $2\leq m\leq k$ and the lemma holds for $m-1$.
We show that $\xi(\cgL)\leq 4(m-1)(\beta_k^4+2\gamma_{k-1}+\delta_{k,m-1})=\delta_{k,m}$.

Since $\cgL\cup\{B_1,\ldots,B_{k-m+1}\}$ is an $m$-fold layer of $\cgF$, there is an $(m-1)$-fold layer $\cgM$ of $\cgF$ such that $\cgL\cup\{B_1,\ldots,B_{k-m+1}\}$ is a layer of $\cgM$.
Let $\cgS$ be the family of those frames in $\cgM$ that intersect all of $B_1,\ldots,B_{k-m+1}$ and do not lie inside any frame in $\cgM$ intersecting all of $B_1,\ldots,B_{k-m+1}$.
It follows that $\omega(\cgS)\leq m-1$ and two frames in $\cgS$ intersect if and only if their filling rectangles intersect.
Since $\cgL\cup\{B_1,\ldots,B_{k-m+1}\}$ is a layer of $\cgM$, for every $F\in\cgL$, there is a frame $E\in\cgM$ that is external to $\cgL\cup\{B_1,\ldots,B_{k-m+1}\}$ and intersects $F$, which implies that $E$ or some frame enclosing $E$ belongs to $\cgS$.
Hence every frame in $\cgL$ intersects or is enclosed in some frame in $\cgS$.

We partition $\cgS$ into four families $\cgS^1$, $\cgS^2$, $\cgS^3$ and $\cgS^4$ so that every frame in $\cgS^1$, $\cgS^2$, $\cgS^3$ or $\cgS^4$ intersects the left, right, bottom or top side of $B_1$, respectively.
Next, we partition $\cgL$ into four families $\cgL^1$, $\cgL^2$, $\cgL^3$ and $\cgL^4$ so that every frame in $\cgL^i$ intersects or is enclosed in some frame in $\cgS^i$.
We have $\xi(\cgL)\leq\xi(\cgL^1)+\xi(\cgL^2)+\xi(\cgL^3)+\xi(\cgL^4)$.
We show that $\xi(\cgL^i)\leq(m-1)(\beta_k^4+2\gamma_{k-1}+\delta_{k,m-1})$.
In the following, we assume without loss of generality that $i=1$, that is, we are to bound $\xi(\cgL^1)$.

We partition $\cgL^1$ into three families $\cgX$, $\cgY$ and $\cgZ$ as follows.
Fix a frame $F\in\cgL^1$.
If there is a frame in $\cgS^1$ that encloses $F$, then we put $F$ in $\cgX$.
Otherwise, if there is a frame in $\cgS^1$ that encloses the entire top or bottom side of $F$, then we put $F$ in $\cgY$.
If neither of the above holds, then we put $F$ in $\cgZ$.

The intersection graph of $\cgS^1$ is isomorphic to the intersection graph of the filling rectangles of the frames in $\cgS^1$, so it is an interval graph, as all these rectangles intersect the vertical line containing the left side of $B_1$.
Hence $\chi(\cgS^1)=\omega(\cgS^1)\leq m-1$.
Therefore, we can partition $\cgS^1$ into $m-1$ families $\cgS_1,\ldots,\cgS_{m-1}$ each consisting of frames with pairwise disjoint filling rectangles.
We also partition each of $\cgX$, $\cgY$ and $\cgZ$ into $m-1$ families $\cgX_1,\ldots,\cgX_{m-1}$, $\cgY_1,\ldots,\cgY_{m-1}$ and $\cgZ_1,\ldots,\cgZ_{m-1}$, respectively, so that
\begin{itemize}
\item if $F\in\cgX_j$, then $F$ is enclosed in a frame in $\cgS_j$,
\item if $F\in\cgY_j$, then the bottom or top side of $F$ is enclosed in a frame in $\cgS_j$,
\item if $F\in\cgZ_j$, then $F$ intersects a frame in $\cgS_j$ and neither of the above holds.
\end{itemize}
We show that $\xi(\cgX_j)\leq\delta_{k,m-1}$, $\xi(\cgY_j)\leq 2\gamma_{k-1}$ and $\xi(\cgZ_j)\leq\beta_k^4$ for $1\leq j\leq m-1$.
Once this is achieved, we will have
\begin{equation*}
\xi(\cgL^1)\leq\sum_{j=1}^{m-1}\xi(\cgX_j)+\sum_{j=1}^{m-1}\xi(\cgY_j)+\sum_{j=1}^{m-1}\xi(\cgZ)\leq(m-1)(\beta_k^4+2\gamma_{k-1}+\delta_{k,m-1}).
\end{equation*}

First, we show that $\xi(\cgX_j)\leq\delta_{k,m-1}$.
We partition $\cgX_j$ into $|\cgS_j|$ families $\cgX_S$ for $S\in\cgS_j$ so that if $F\in\cgX_S$, then $F$ is enclosed in $S$.
The families $\cgX_S$ for $S\in\cgS_j$ are mutually disjoint, as the filling rectangles of the $S\in\cgS_j$ are pairwise disjoint.
Let $S\in\cgS_j$.
It follows that $S$ intersects all of $B_1,\ldots,B_{k-m+1}$.
Since $\cgX_S\cup\{B_1,\ldots,B_{k-m+1},S\}$ is a subfamily of $\cgM$, it is an $(m-1)$-fold layer of $\cgF$.
The induction hypothesis yields $\xi(\cgX_S)\leq\delta_{k,m-1}$.
Hence $\xi(\cgX_j)=\max_{S\in\cgS_j}\xi(\cgX_S)\leq\delta_{k,m-1}$.

Now, we show that $\xi(\cgY_j)\leq 2\gamma_{k-1}$.
We partition $\cgY_j$ into two families $\cgY_j^b$ and $\cgY_j^t$ so that if $F\in\cgY_j^b$, then the bottom side of $F$ is enclosed in some frame in $\cgS_j$, while if $F\in\cgY_j^t$, then the top side of $F$ is enclosed in some frame in $\cgS_j$.
We show that $\omega(\cgY_j^b)\leq k-1$ and $\omega(\cgY_j^t)\leq k-1$.
Then, it will follow that $\xi(\cgY_j)\leq\xi(\cgY_j^b)+\xi(\cgY_j^t)\leq 2\gamma_{k-1}$.

To see that $\omega(\cgY_j^b)\leq k-1$, let $\cgK$ be a clique in $\cgY_j^b$.
Let $F$ be the frame in $\cgK$ with maximum $\bottomside(F)$, and let $S$ be the frame in $\cgS_j$ enclosing the bottom side of $F$.
It follows that $F$ and $S$ intersect.
Moreover, every other frame in $\cgK\setminus\{F\}$ intersects or encloses the bottom side of $F$ and thus intersects $S$ as well.
Hence $\cgK\cup\{S\}$ is a clique in $\cgF$.
This implies $|\cgK|\leq\omega(\cgF)-1\leq k-1$.
The proof that $\omega(\cgY_j^t)\leq k-1$ is analogous.

Finally, we show that $\xi(\cgZ_j)\leq\beta_k^4$.
The definition of $\cgZ_j$ and the assumption that $\cgF$ is non-crossing imply that the right side of no frame in $\cgZ_j$ intersects or is enclosed in any frame in $\cgS_j$.
For each frame $F\in\cgZ_j$, we define four curves, the \emph{short lower}, \emph{short upper}, \emph{long lower} and \emph{long upper trace} of $F$, as follows.
The short lower (upper) trace of $F$ starts at the bottom (top) right corner of $F$ and follows along the bottom (top) side of $F$ and possibly further along $F$ until it reaches an intersection point with a frame $S\in\cgS_j$ on either the bottom (top) or the left side of $F$.
The long lower (upper) trace of $F$ starts at the top (bottom) right corner of $F$, follows along the right side of $F$ to the bottom (top) corner of $F$, and then continues along the short lower (upper) trace until the intersection point with a frame $S\in\cgS_j$.
See Figure \ref{fig:traces} for an illustration.

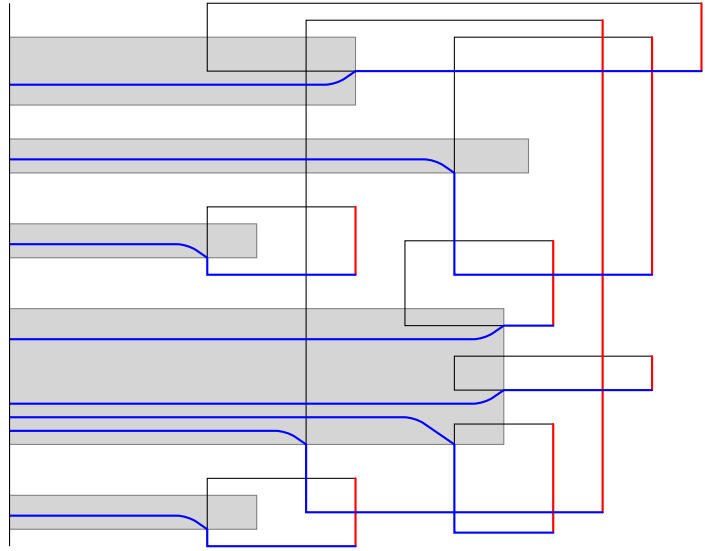
\begin{figure}[t]
\centering
\begin{tikzpicture}[xscale=0.65,yscale=0.45]
  \fill[gray!33!white] (0,0.5) rectangle (5,1.5);
  \fill[gray!33!white] (0,3) rectangle (10,7);
  \fill[gray!33!white] (0,8.5) rectangle (5,9.5);
  \fill[gray!33!white] (0,11) rectangle (10.5,12);
  \fill[gray!33!white] (0,13) rectangle (7,15);
  \draw[gray] (0,0.5)--(5,0.5)--(5,1.5)--(0,1.5);
  \draw[gray] (0,3)--(10,3)--(10,7)--(0,7);
  \draw[gray] (0,8.5)--(5,8.5)--(5,9.5)--(0,9.5);
  \draw[gray] (0,11)--(10.5,11)--(10.5,12)--(0,12);
  \draw[gray] (0,13)--(7,13)--(7,15)--(0,15);
  \draw (4,0)--(4,2)--(7,2);
  \draw (6,3)--(6,15.5)--(12,15.5);
  \draw (9,3)--(9,3.6)--(11,3.6);
  \draw (10,4.6)--(9,4.6)--(9,5.6)--(13,5.6);
  \draw (10,6.5)--(8,6.5)--(8,9)--(11,9);
  \draw (4,8.5)--(4,10)--(7,10);
  \draw (9,11)--(9,15)--(13,15);
  \draw (7,14)--(4,14)--(4,16)--(14,16);
  \draw[red,thick,line cap=round] (7,0)--(7,2);
  \draw[red,thick,line cap=round] (12,1)--(12,15.5);
  \draw[red,thick,line cap=round] (11,0.4)--(11,3.6);
  \draw[red,thick,line cap=round] (13,4.6)--(13,5.6);
  \draw[red,thick,line cap=round] (11,6.5)--(11,9);
  \draw[red,thick,line cap=round] (7,8)--(7,10);
  \draw[red,thick,line cap=round] (13,8)--(13,15);
  \draw[red,thick,line cap=round] (14,14)--(14,16);
  \draw[blue,thick,line cap=round] (4,0.5)--(4,0)--(7,0);
  \draw[blue,thick,line cap=round] (6,3)--(6,1)--(12,1);
  \draw[blue,thick,line cap=round] (9,3)--(9,0.4)--(11,0.4);
  \draw[blue,thick,line cap=round] (10,4.6)--(13,4.6);
  \draw[blue,thick,line cap=round] (10,6.5)--(11,6.5);
  \draw[blue,thick,line cap=round] (4,8.5)--(4,8)--(7,8);
  \draw[blue,thick,line cap=round] (9,11)--(9,8)--(13,8);
  \draw[blue,thick,line cap=round] (7,14)--(14,14);
  \draw[blue,thick,rounded corners] (0,0.9)--(3.6,0.9)--(4,0.5);
  \draw[blue,thick,rounded corners] (0,3.4)--(5.6,3.4)--(6,3);
  \draw[blue,thick,rounded corners] (0,3.8)--(8.2,3.8)--(9,3);
  \draw[blue,thick,rounded corners] (0,4.2)--(9.6,4.2)--(10,4.6);
  \draw[blue,thick,rounded corners] (0,6.1)--(9.6,6.1)--(10,6.5);
  \draw[blue,thick,rounded corners] (0,8.9)--(3.6,8.9)--(4,8.5);
  \draw[blue,thick,rounded corners] (0,11.4)--(8.6,11.4)--(9,11);
  \draw[blue,thick,rounded corners] (0,13.6)--(6.6,13.6)--(7,14);
  \draw (0,0)--(0,16);
\end{tikzpicture}
\caption{Short (blue) and long (blue${}+{}$red) lower traces}
\label{fig:traces}
\end{figure}

All short lower (upper) traces can be connected to the left side of $B_1$ by pairwise disjoint $x$-monotone curves inside the frames in $\cgS_j$, thus forming a family of $x$-semimonotone curves with the same intersection graph.
Any two of these curves intersect in at most one point, because so do any two short lower (upper) traces.
Therefore, by Theorem \ref{thm:suk}, there are proper colorings $\phi_\ell$ and $\phi_u$ of all short lower and short upper traces, respectively, with $\beta_k$ colors.
To prove that $\xi(\cgZ_j)\leq\beta_k^4$, it is enough to show that $\xi(\cgZ')\leq\beta_k^2$ for any family $\cgZ'\subset\cgZ_j$ of frames whose short lower traces have the same color in $\phi_\ell$ and whose short upper traces have the same color in $\phi_u$.

Let $\cgZ'$ be such a family.
Since the short lower (upper) traces of any two frames in $\cgZ'$ are disjoint, their long lower (upper) traces intersect in at most one point.
Consequently, the same argument as for short traces yields proper colorings $\psi_\ell$ and $\psi_u$ of all long lower and long upper traces, respectively, with $\beta_k$ colors.
The intersection of any two frames in $\cgZ'$ whose lower long traces are disjoint and whose upper long traces are disjoint is leftward-directed.
Therefore, the frames in $\cgZ'$ whose long lower traces have the same color in $\psi_\ell$ and whose long upper traces have the same color in $\psi_u$ form a leftward-directed family.
This shows that $\xi(\cgZ')\leq\beta_k^2$.
\end{proof}

\begin{lemma}
\label{lem:geometry-outer}
Every non-crossing family of frames\/ $\cgF$ with\/ $\omega(\cgF)\leq k$ satisfies\/ $\xi(\cgF)\leq\gamma_k$.
\end{lemma}

\begin{proof}
The proof goes by induction on $k$.
The statement is trivial for $k=1$, so assume that $k\geq 2$ and the statement holds for $k-1$.
Let $\cgF$ be a non-crossing family of frames with $\omega(\cgF)\leq k$.
By Corollary \ref{cor:bfs}, $\cgF$ has a partition $\scrP$ into $k$-fold layers, and moreover, $\scrP$ can be partitioned into $2^k$ subclasses $\scrP_1,\ldots,\scrP_{2^k}$ each consisting of mutually disjoint $k$-fold layers.
We show that $\xi(\cgL)\leq\alpha_k\delta_{k,k}$ for any $k$-fold layer $\cgL\in\scrP$.
Once this is done, we will have $\xi(\bigcup\scrP_i)\leq\alpha_k\delta_{k,k}$ (because each $\scrP_i$ consists of mutually disjoint $k$-fold layers) and thus
\begin{equation*}
\xi(\cgF)\leq\sum_{i=1}^{2^k}\xi({\textstyle\bigcup\scrP_i})\leq 2^k\alpha_k\delta_{k,k}=\gamma_k.
\end{equation*}

Let $\cgL\in\scrP$.
Let $\cgR$ be the family of those frames in $\cgL$ that do not lie inside any frame in $\cgL$.
Every frame in $\cgL\setminus\cgR$ is enclosed in some frame in $\cgR$.
Hence $\cgL\setminus\cgR$ can be partitioned into $|\cgR|$ families $\cgL_R$ for $R\in\cgR$ so that every frame in $\cgL_R$ is enclosed in $R$.
Since $\cgL_R\cup\{R\}$ is a $k$-fold layer of $\cgF$, it follows from Lemma \ref{lem:geometry-inner} that $\xi(\cgL_R)\leq\delta_{k,k}$.
The intersection graph of $\cgR$ is isomorphic to the intersection graph of the filling rectangles of the frames in $\cgR$, so Theorem \ref{thm:rect} yields $\chi(\cgR)\leq\alpha_k$.
Hence $\cgR$ can be partitioned into $\alpha_k$ subfamilies $\cgR_1,\ldots,\cgR_{\alpha_k}$ each consisting of frames with pairwise disjoint filling rectangles.
For $1\leq j\leq\alpha_k$, the families $\cgR_j$ and $\cgL_R$ for $R\in\cgR_j$ are mutually disjoint, so
\begin{equation*}
\xi\Bigl(\cgR_j\cup\bigcup_{R\in\cgR_j}\cgL_R\Bigr)=\max\bigl\{\xi(\cgR_j),\max_{R\in\cgR_j}\xi(\cgL_R)\bigr\}\leq\delta_{k,k}.
\end{equation*}
Therefore,
\begin{equation*}
\xi(\cgL)\leq\sum_{j=1}^{\alpha_k}\xi\Bigl(\cgR_j\cup\bigcup_{R\in\cgR_j}\cgL_R\Bigr)\leq\alpha_k\delta_{k,k}.\qedhere
\end{equation*}
\end{proof}

\begin{proof}[Proof of Lemma \ref{lem:frame-geometry}]
Let $\cgF$ be a family of frames with $\omega(\cgF)\leq k$.
As it has been explained at the beginning of this section, it is enough to bound $\xi(\cgF)$ in terms of $k$.
By Lemma \ref{lem:cross}, we can partition $\cgF$ into $k$ non-crossing families $\cgF_1,\ldots,\cgF_k$.
By Lemma \ref{lem:geometry-outer}, we have $\xi(\cgF_i)\leq\gamma_k$ for $1\leq i\leq k$.
Hence $\xi(\cgF)\leq\sum_{i=1}^k\xi(\cgF_i)\leq k\gamma_k$.
\end{proof}

\begin{proof}[Proof of Lemma \ref{lem:frame-clean}]
Let $\cgF$ be a triangle-free family of frames.
By Lemma \ref{lem:bfs}, $\cgF$ has a partition $\scrP$ into layers, which can be further split as $\scrP=\scrP_1\cup\scrP_2$ so that each of $\scrP_1$ and $\scrP_2$ consists of mutually disjoint layers.
We are going to show that each layer in $\scrP$ is a clean family of frames.
This will show that the subfamilies $\bigcup\scrP_1$ and $\bigcup\scrP_2$ of $\cgF$ satisfy the conclusion of the~lemma.

Choose any layer $\cgL\in\scrP$.
Suppose there are three frames $F_1,F_2,F_3\in\cgL$ such that $F_1$ and $F_2$ intersect and $F_3$ is enclosed in both $F_1$ and $F_2$.
By the definition of a layer, there is a frame $E\in\cgF\setminus\cgL$ external to $\cgL$ and intersecting $F_3$.
Clearly, $E$ must intersect both $F_1$ and $F_2$, thus creating a triangle in the intersection graph of $\cgF$.
This contradiction shows that $\cgL$ is indeed a clean family of frames.
\end{proof}

\section{Problems}

The authors of \cite{PKK+14} asked for the maximum possible chromatic number of a triangle-free intersection graph of $n$ segments.
In this paper, we resolved a similar question for frames.
The following problems ask whether segment graphs behave similarly to frame graphs with respect to containment of overlap game graphs.

\begin{problem}
Can every triangle-free segment intersection graph be decomposed into a bounded number of overlap game graphs?
\end{problem}

\begin{problem}
Does every triangle-free segment intersection graph with chromatic number $k$ contain an overlap game graph with chromatic number at least $ck$ as an induced subgraph, for some absolute constant $c>0$?
\end{problem}

The positive answer to either question would yield the bound of $\Theta(\log\log n)$ on the chromatic number of triangle-free segment intersection graphs, while the negative answer would help us understand the limitations of our methods.

Most results of this paper concern triangle-free graphs, but similar statements are likely to hold for $K_k$-free graphs.

\begin{problem}
What is the maximum possible chromatic number of a $K_k$-free overlap graph of $n$ rectangles (intersection graph of $n$ frames)?
\end{problem}

\bibliographystyle{plain}
\bibliography{loglog}

\end{document}